\documentclass[11pt,letterpaper]{article}

\usepackage{amsmath,amstext,amssymb,amsfonts,bbold}
\usepackage{amsmath,amsthm}
\usepackage[shortlabels]{enumitem}
\usepackage{graphicx,psfrag}
\usepackage{comment}
\usepackage{xspace,calc}

\usepackage{fullpage}
\usepackage{hyperref}

\newcommand{\np}{{\em NP}\xspace}
\newcommand{\nphard}{\np-hard\xspace}

\newtheorem{theorem}{Theorem}[section]

\newtheorem{lemma}[theorem]{Lemma}
\newtheorem{claim}[theorem]{Claim}

\newtheorem{conjecture}{Conjecture}

{\theoremstyle{remark} }
{\theoremstyle{definition} \newtheorem{definition}[theorem]{Definition}}

\newenvironment{proofof}[1]{\begin{proof}[Proof of #1]}{\end{proof}}

\newcommand{\mone}{{\mathbb{1}}}

\DeclareMathOperator{\poly}{poly}
\DeclareMathOperator{\polylog}{polylog}
\DeclareMathOperator{\val}{val}
\DeclareMathOperator{\conv}{conv}
\DeclareMathOperator{\opt}{opt}
\DeclareMathOperator{\argmax}{argmax}

\newcommand{\simp}{\ensuremath{\Dt}}

\newcommand{\e}{\ensuremath{\epsilon}}
\newcommand{\eps}{\ensuremath{\epsilon}}
\newcommand{\ve}{\ensuremath{\varepsilon}}
\newcommand{\gm}{\ensuremath{\gamma}}

\newcommand{\ld}{\ensuremath{\lambda}}

\newcommand{\al}{\ensuremath{\alpha}}

\newcommand{\tht}{\ensuremath{\theta}}
\newcommand{\Tht}{\ensuremath{\Theta}}
\newcommand{\dt}{\ensuremath{\delta}}
\newcommand{\Dt}{\ensuremath{\Delta}}
\newcommand{\sg}{\ensuremath{\sigma}}
\newcommand{\Sg}{\ensuremath{\Sigma}}
\newcommand{\w}{\ensuremath{\omega}}
\newcommand{\Om}{\ensuremath{\Omega}}

\newcommand{\sm}{\ensuremath{\setminus}}
\newcommand{\es}{\ensuremath{\emptyset}}
\newcommand{\sse}{\subseteq}
\newcommand{\frall}{\ensuremath{\text{ for all }}}

\newcommand{\R}{\ensuremath{\mathbb R}}
\newcommand{\T}{\ensuremath{\mathcal{T}}}

\newcommand{\A}{\ensuremath{\mathcal{A}}}
\newcommand{\B}{\ensuremath{\mathcal{B}}}
\newcommand{\C}{\ensuremath{\mathcal{C}}}

\newcommand{\F}{\ensuremath{\mathcal F}}
\newcommand{\G}{\ensuremath{\mathcal{G}}}

\newcommand{\I}{\ensuremath{\mathcal I}}

\newcommand{\Rc}{\ensuremath{\mathcal R}}

\newcommand{\hmu}{\ensuremath{\hat\mu}}
\newcommand{\tx}{\ensuremath{\tilde x}}
\newcommand{\tmu}{\ensuremath{\tilde\mu}}
\newcommand{\bA}{\ensuremath{\overline{A}}}
\newcommand{\y}{\ensuremath{y}}
\newcommand{\x}{\ensuremath{x}}
\newcommand{\OPT}{\ensuremath{\mathit{OPT}}}

\newcommand{\Pd}{\ensuremath{\text{\eqref{primd}}}}
\newcommand{\Dd}{\ensuremath{\text{\eqref{duald}}}}
\newcommand{\pclique}{\ensuremath{\mathbf{PClique}}}
\newcommand{\NE}{\ensuremath{\mathit{NE}}}
\newcommand{\poa}{PoA\xspace}
\newcommand{\brf}{\ensuremath{\tilde f}}
\newcommand{\bidens}{\ensuremath{\text{bi-density}}}
\newcommand{\Row}{\ensuremath{\text{\bf Row}}\xspace}
\newcommand{\Col}{\ensuremath{\text{\bf Col}}\xspace}

\newcommand{\set}[1]{\left\{ #1 \right\}}

\def\pr{\qopname\relax n{Pr}}

\def\min{\qopname\relax n{min}}
\def\max{\qopname\relax n{max}}

\def\Pr{\qopname\relax n{\mathbf{Pr}}}
\def\Ex{\qopname\relax n{\mathbf{E}}}

\def\T{\mathcal{T}}

\newcommand{\eat}[1]{}

\newenvironment{lp*}{\begin{equation*}  \begin{array}{lll}}{\end{array}\end{equation*}}

\title{Hardness Results for Signaling in \\ Bayesian Zero-Sum and Network Routing Games}

\author{
  Umang Bhaskar\thanks{Tata Institute of Fundamental Research, Mumbai, India.
  Part of this work was done while the author
  was a postdoctoral scholar at the University of Waterloo and was supported in part
  by Chaitanya Swamy's research grants. Email: {\tt umang@tifr.res.in}.}
\and
  Yu Cheng\thanks{University of Southern California, CA 90089, USA.
  Supported in part by Shang-Hua Teng's Simons Investigator Award from the Simons Foundation.
  Email: {\tt yu.cheng.1@usc.edu}.}
\and
  Young Kun Ko\thanks{Princeton University, NJ 08544, USA. Email: {\tt yko@cs.princeton.edu}.}
\and
  Chaitanya Swamy\thanks{Department of Combinatorics and Optimization, University of Waterloo,
  Waterloo, ON N2L 3G1, Canada. Supported in part by NSERC grant
  32760-06, an NSERC Discovery Accelerator Supplement Award, and an Ontario Early
  Researcher Award.
  Email: {\tt cswamy@uwaterloo.ca}.}
}

\date{\today}

\begin{document}

\maketitle

\begin{abstract}
We study the optimization problem faced by a perfectly informed principal in a Bayesian
game, who reveals information to the players about the state of nature to obtain a desirable
equilibrium. This {\em signaling problem} is the natural design question
motivated by uncertainty in games and has attracted much recent attention.
We present new hardness results for signaling problems in
(a) Bayesian two-player zero-sum games, and (b) Bayesian network routing games.

For {\em Bayesian zero-sum games}, when the principal seeks to maximize the equilibrium
utility of a player, we show that it is \nphard to obtain an additive FPTAS.
Our hardness proof exploits {\em duality} and the equivalence of separation and
optimization in a novel way.
Further, we rule out an additive PTAS assuming {\em planted clique} hardness,
  which states that no polynomial time algorithm can recover a planted clique
  from an Erd\H{o}s-R\'enyi random graph.
Complementing these, we obtain a PTAS for a structured class of zero-sum games
  (where obtaining an FPTAS is still \nphard) when the payoff matrices obey a Lipschitz condition.
Previous results ruled out an FPTAS assuming planted-clique hardness,
  and a PTAS only for {\em implicit} games with {\em quasi-polynomial-size strategy sets}.

For {\em Bayesian network routing games}, wherein the principal seeks to minimize the
average latency of the Nash flow, we show that it is \nphard to obtain a
(multiplicative) $\bigl(\frac{4}{3}-\e\bigr)$-approximation, even for linear latency functions.
This is the {\em optimal} inapproximability result for linear latencies, since
we show that full revelation achieves a $\frac{4}{3}$-approximation for linear latencies.
\end{abstract}

\section{Introduction} \label{intro}

\def\xx{\mathbf{x}}
\def\xxappr{\tilde{\mathbf{\xx}}}

\def\calG{\mathcal{G}}
\def\Nash{\mathrm{NashEquilibrium}}

In Bayesian games,
players' payoffs depend on the state of nature, which
may be hidden from the players. Instead, players receive a {\em signal} regarding the state
of nature which they use to form beliefs about their payoffs, and choose their
strategies. Thus the strategic decisions and payoffs of the players depend crucially on
the information available from the signal they receive. Since applications are often
rife with uncertainty, understanding the effect of information available to players is
a fundamental problem in game theory; see,
e.g.,~\cite{blackwell51,akerloflemons,hirshleifer71,MW82,lehrermediation,bergemanndiscrimination}.
Whereas for a single player, it is known that more
information leads to better payoffs~\cite{blackwell51},
with multiple players, outcomes are more complex and often counterintuitive with
``more'' (information) not necessarily translating to ``better'' (payoffs).
The latter was first observed by Hirshleifer \cite{hirshleifer71};
recently, Dughmi \cite{Dughmi14} gave an example where neither full-revelation nor
no-revelation is optimal.

While classical work has focused on the role of information in influencing
strategies,
the computational problem of \emph{designing} optimal
information structures for Bayesian games, commonly called the {\em signaling problem},
has received much recent attention
~\cite{miltersensignals,DIR14,emeksignaling,guo}.
Here, a perfectly-informed principal seeks to
reveal selective information to the players to optimize some function of the resulting
equilibrium, such as the revenue, or payoff of a particular player.
Two-player zero-sum games and network routing games are natural starting points
for investigating the signaling problem due to their fundamental importance
and appealing structure. They admit a canonical, tractable choice of equilibrium;
this also decouples the concerns of optimal-signaling computation and equilibrium computation.

\paragraph{Our results} We study signaling in two widely studied classes of games:
two-player zero-sum games, and network routing games.
As in much of previous work, in our setting players share the
same prior belief on the state of nature, and signaling schemes are symmetric:
the principal reveals the same information to all players. Further, as previously, our results are for \emph{additive} approximations in Bayesian zero-sum games, and for \emph{multiplicative} approximations in Bayesian network routing games.
Our main contribution is to derive hardness results for these classes of games
that close the gap between what is achievable in polytime (or
quasi-polytime) and what is intractable.

In Section~\ref{zerosum}, we consider {\em Bayesian (two-player) zero-sum games},
in which the principal seeks to maximize the value of the game --- the
equilibrium payoff of the row player.\footnote{In zero-sum games, this also captures the
problem of maximizing a weighted combination of players' equilibrium payoffs.}
First, we settle the complexity of the signaling problem with respect to \nphard{}ness
by showing that it is \nphard to obtain an additive FPTAS (Theorem~\ref{nofptasthm}).
Previous work by Dughmi \cite{Dughmi14} ruled out an FPTAS
assuming the planted clique hardness (see Conjecture \ref{pcliquehard}).
Thus, we replace an {\em average-case} hardness
assumption with the much more conventional worst-case assumption of \nphard{}ness.

Next, we consider the hardness of obtaining a PTAS for the signaling problem.
Since there is a quasi-polytime approximation scheme for signaling given by Cheng et al.~\cite{CCDEHT15},
it is unlikely that a PTAS for signaling is \nphard{}.
We show that assuming planted-clique hardness,
there does not exist a PTAS for the signaling problem (Theorem~\ref{noptasthm}).
Previously, the non-existence of a PTAS was shown (assuming
planted-clique hardness) only for {\em implicit} zero-sum games with
{\em quasi-polynomial-size strategy sets}~\cite{Dughmi14}.
Complementing these hardness results, we devise a PTAS
for a structured class of Bayesian zero-sum games (Theorem~\ref{ptas}),
when the payoff matrices obey a Lipschitz condition.

In Section~\ref{sec:congestion}, we consider the signaling problem in (nonatomic, selfish)
{\em Bayesian network routing games}, wherein the principal seeks to reveal partial information
to {\em minimize} the average latency of the equilibrium flow.
We show that it is \nphard to obtain any multiplicative approximation better than
$\frac{4}{3}$, even with linear latency functions (Theorem~\ref{thm:routinghardness}).
This yields an {\em optimal} inapproximability result for linear latencies,
since we show that full revelation obtains the {\em price of anarchy} of the routing game
as its approximation ratio (Theorem~\ref{fullrev}), which is $\frac{4}{3}$ for linear
latency functions~\cite{RT02}.
These are the {\em first} results for the complexity of
signaling in Bayesian network routing games.

We also obtain hardness results
for two related signaling problems in Bayesian zero-sum games (Section~\ref{extn}).
Firstly, we rule out a PTAS for computing the best prior (the {\em maximum prior problem}),
under the {\em exponential time hypothesis} (ETH).
Previously, \cite{CCDEHT15} studied a mixture-selection problem and showed that in the absence of a property called \emph{noise-stability},
  obtaining a PTAS was hard, assuming planted-clique hardness.
Our result shows that in their setting a QPTAS is in fact the best possible approximation obtainable, assuming the ETH.
Finally, if the principal's value depends on the players' strategies, and not just their payoffs,
  we show that obtaining a PTAS is \nphard (Theorem~\ref{maxposthm}).

\paragraph{Our techniques}
Our results for Bayesian zero-sum games are obtained via two main ideas.
Our \nphard{}ness result, the PTAS for a structured class of games, and the PTAS-hardness
for the maximum prior problem, all follow by
considering the signaling problem from a {\em dual} perspective.
The signaling problem can be written as a mathematical program \eqref{primal} with
linear objective and constraints, but an infinite number of variables.
Ignoring
this issue,
we can consider the dual problem \eqref{dual}.
Motivated by the separation problem for the dual, we consider the
{\em dual signaling  problem} (Section~\ref{dsp}).
Our key insight is that {\em the dual signaling problem is a rather
useful tool for both deriving hardness results and devising approximation algorithms}.
This usefulness stems from the equivalence of separation and
optimization~\cite{ellipsoidbook}, which shows that an algorithm for the separation
problem can be used to solve the optimization problem and vice versa. We exploit and build
upon this equivalence.
We prove that this equivalence holds despite the
infinite-dimensionality of
\eqref{primal}, and furthermore,
is approximation preserving: an FPTAS for signaling yields
an FPTAS for the dual signaling problem (Theorem~\ref{opt2sep}), and a
PTAS for the dual signaling problem
yields a PTAS for signaling (Theorem~\ref{sep2opt}).

This equivalence paves the way for our results.
Whereas, typically, an (approximate) separation
oracle is used to (approximately) solve the optimization problem, we exploit this
equivalence in an unorthodox fashion by also leveraging the {\em hardness} of the dual
signaling problem to prove hardness results for the signaling (i.e., primal optimization)
problem.
We show that it is \nphard to obtain an FPTAS for the dual signaling problem,
and thus obtain that it is \nphard to obtain an FPTAS for the signaling
problem.
Notably, in contrast to the (weaker) planted-clique hardness result
in~\cite{Dughmi14} for the signaling problem, we obtain our
\nphard{}ness result with minimal effort, a fact that underscores the benefits
of moving to the dual signaling problem.

On the positive side, we obtain a PTAS for the dual signaling problem for our structured
class of Bayesian zero-sum games,
which thus yields a PTAS for the signaling problem for this class.
Interestingly, when cast in the mixture-selection framework of Cheng et al.~\cite{CCDEHT15}, the
signaling problem for our structured class {\em does not} satisfy the noise-stability
property stated in~\cite{CCDEHT15}.
In the absence of noise-stability, \cite{CCDEHT15} showed planted-clique hardness for obtaining a PTAS.
Our result bypasses this hardness result, and obtains a PTAS for a problem for which noise-stability does not hold.
Finally, we show that a PTAS for the maximum-prior problem yields a PTAS for the dual
signaling problem, and we rule out the latter via a simple, clean reduction from the
{\em best-Nash} problem and the recent result of Braverman et al.~\cite{BKW15}.
This result also strengthens the hardness result from~\cite{CCDEHT15} mentioned above, by showing that in the absence of noise-stability, a QPTAS is the best-possible approximation for the mixture selection problem, assuming the ETH.

Our second main idea, used to rule out a PTAS assuming planted-clique hardness, is a
``direct'' reduction that combines and strengthens techniques from~\cite{Dughmi14,FNS07}.
We utilize the {\em planted clique cover} problem defined
in~\cite{Dughmi14} ---
multiple cliques are now planted and one seeks to recover a constant fraction of them ---
and shown to be at least as hard as the planted clique problem.
The idea is to set up a Bayesian zero-sum game
where both the principal and the row player must randomize over $\Omega(\log n)$-size
high-density node sets for the signaling scheme to achieve large value;
recovering these large-density sets from a near-optimal signaling scheme allows
one to solve the planted-clique cover (and hence, the planted clique) problem.
The FPTAS-hardness reduction in~\cite{Dughmi14} creates a network security game
(see Section~\ref{prelim})
with payoffs of absolute value $\Omega(\log^2 n)$ (or alternatively, a
quasi-polynomial-size strategy set for the column player) to enforce the above property.
Payoffs of magnitude $\Omega(\log n)$ seem necessary with this kind of approach, which
therefore only yields an  $O({1}/{\log n})$ gap that is insufficient to rule out a PTAS.
We abandon the use of network security games and instead leverage a device
from~\cite{FNS07} to ensure the above ``large-spreading'' property.
This idea is also used
to show planted-clique hardness for the {\em best-Nash} problem~\cite{hazankrauthgamer};
however
we are constrained to work with zero-sum games, and therefore need to apply this
idea carefully.
A subtle, but crucial, technical issue is that we need to significantly tighten the
planted-clique recovery result in~\cite{Dughmi14}.
To recover a specific planted clique $S$
of size $k=\omega(\log^2 n)$ (in the presence of other such planted cliques),
\cite{Dughmi14} requires a set $T$ with $|T|=\Theta(k),\ |S\cap T|=\Omega(k)$,
whereas we only require that $|T|, |S\cap T|=\Omega(\log n)$, and this is crucial since we
can only ensure that spreading takes place over $O(\log n)$-size sets.

Our hardness result for Bayesian routing games
is a direct reduction
from the problem of computing edge tolls that minimize the total (latency + toll)-cost of
the resulting equilibrium flow, which
is inapproximable within a factor of $\frac{4}{3}$.

\paragraph{Related work}
Whereas
understanding the role of information in influencing strategies is
a classical
problem in game theory,
the computational problem of designing optimal information structures
has been studied more recently.
Much of this work has focused on signaling in auctions, where the goal is to
maximize revenue~\cite{emeksignaling,miltersensignals,guo} or social
welfare~\cite{DIR14}.
Dughmi \cite{Dughmi14} initiated the computational study of signaling in Bayesian zero-sum
games, and obtained various hardness results under the planted-clique hardness assumption.
This work left open the question of whether hardness results can be obtained under
standard worst-case assumptions, such as {\em P}$\neq$\np, a question that we answer in
the affirmative.
On the positive side, Cheng et al.~\cite{CCDEHT15}
showed that for
Bayesian normal form games with a constant number of players and for general objectives of
the principal, an $\eps$-approximate signaling scheme that maximizes the objective
at an $\eps$-approximate Nash equilibrium can be computed in quasi-polynomial time.
This work left open the question of whether a PTAS is possible for signaling in Bayesian
zero-sum games.
We preclude this
under the planted-clique hardness
assumption, and complementing this, design a PTAS for a structured class of games.
As noted earlier, the latter result does not follow from~\cite{CCDEHT15}
since the resulting signaling problem fails to have small noise stability.

The planted-clique problem
was introduced by
Jerrum~\cite{jerrum92} and Ku\u{c}era~\cite{kucera95}, and
despite extensive efforts (see, e.g.,~\cite{amesvavasis,feigeron,dekel} and the references therein),
no polytime algorithm is known for recovering cliques of size
$k=o(\sqrt{n})$. There is a quasi-polytime algorithm known
when $k\geq 2\log_2 n$; on the other hand,
various algorithmic strategies have been ruled
out for this problem~\cite{jerrum92,feigeprobable,feldmanclique}.
The planted-clique problem has thus been used in various reductions (see,
e.g.,~\cite{hazankrauthgamer,juels}), and is an example where an
{\em average-case hardness assumption} has been used to derive hardness results.

Recently, Rubinstein \cite{rubinstein15} has independently also obtained hardness results
for signaling in zero-sum games.
He shows that there is no additive PTAS assuming ETH,
  and obtaining a multiplicative PTAS is \nphard.
These results are orthogonal to ours, as there is no known
  reduction between ETH and planted-clique hardness.
Further, \nphard{}ness of a multiplicative PTAS does not rule out an additive FPTAS.

In Bayesian network routing games,
  \cite{VassermanFH15} study the ability of signaling to reduce the average latency.
They define the \emph{mediation ratio} as the average latency at equilibrium for the best (private) signaling scheme,
  to the average latency for the social optimum, and give tight bounds on the mediation ratio with graphs consisting of parallel links.
On these simple networks, navigation services (such as Waze or Google Maps) cannot do anything
  to improve the latency of the Nash flow.
Our work, in contrast, studies the \emph{computational complexity} of obtaining the best (public) signaling scheme in general graphs,
  and conclude that finding an $\bigl(\frac{4}{3}-\e\bigr)$-approximation is \nphard.

\section{Preliminaries and notation} \label{prelim}

We use $\R_+$ for the set of nonnegative reals.
For integer $n$, $[n] := \{1,2,\ldots,n\}$.
If $n\geq 1$,  we use $\simp_n$ to denote the $(n-1)$-dimensional simplex $\{x \in \mathbb{R}_+^n: \sum_i x_i = 1\}$.
Let $\mone_n\in\R^n$ be the vector with 1 in all its entries,
$I_{n \times n}$ be the $n \times n$ identity matrix, and $e_i$ be the vector containing
1 as its $i$-th entry and 0 elsewhere.

\paragraph{Bayesian zero-sum games and signaling schemes}
A {\em Bayesian zero-sum game} is specified by a tuple
$\bigl(\Tht,\{\A^\theta\}_{\theta \in \Theta},\lambda\bigr)$,
where $\Theta = \{1, \dots, M\}$ denotes the states of nature,
and $\lambda$ is a prior distribution on the states of nature (thus $\ld\in\simp_M$).
We assume the row and column player has $r, c$ pure strategies respectively.
For each state of nature $\theta \in \Tht$,
$\A^\theta\in[-1,1]^{r \times c}$ specifies the payoffs of the row player in a zero-sum game.
Let $\mu \in \simp_M$ be an arbitrary distribution over states of nature.
Then $\A^\mu := \sum_{\theta \in \Theta} \mu_\theta \A^\theta$ is the matrix of expected payoffs
for the row player under distribution $\mu$.

A {\em signaling scheme} is a policy by which a principal reveals
  (partial) information about the state of nature.
We focus on symmetric signaling schemes which reveals the same
  information to all the players.
A signaling scheme specifies a set of signals $\Sg$ and a map
  $\varphi:\Tht \mapsto \simp_{|\Sg|}$ from the states of nature $\Tht$
  to distributions over the signals in $\Sg$.
Thus, $\varphi(\tht)_\sg$ is the probability that the principal selects signal $\sg$ when the state of nature is $\tht$.
When the state $\tht$ is revealed, the principal computes a signal
  $\sigma \sim \varphi(\tht)$.
Both players receive $\sg$ and correspondingly update their belief on the
  state-distribution to $\mu^\sg$, where for each state $\tht$,
\[
\mu^\sg_\tht ~ = ~ \frac{\pr(\sg | \tht) \, \pr(\tht)}{\pr(\sg)} ~ = ~ \frac{\varphi(\tht)_\sg \lambda_\tht}{\sum_{\tht' \in \Tht} \varphi(\tht')_\sg} \, .
\]
\noindent The players then, based on their posterior belief, play the zero-sum game given by $\A^{\mu^\sg}$.

Each signal $\sg$ thus yields a posterior distribution $\mu^\sg \in \simp_M$, and these posterior distributions form a convex decomposition of the prior $\lambda = \sum_\sg \pr(\sg) \mu^\sg$.
As observed in~\cite{Dughmi14}, specifying a signaling scheme $(\Sigma, \varphi)$ is in fact equivalent to specifying a distribution $\alpha$ over posterior distributions $\mu \in \simp_M$ that yield a convex decomposition of the prior $\ld$. Thus, a signaling scheme can also be described as $\al := (\al_\mu)_{\mu \in \simp_M}$, where $\sum_{\mu \in \simp_M} \al_\mu \mu = \ld$.
The signals $\Sigma$ in such a signaling scheme are described implicitly, and correspond to the posteriors $\mu$ for which $\al_\mu > 0$.
This will be our perspective on signaling schemes throughout.
In Section~\ref{noptas}, we will explicitly need to describe the signals, and then use $\mu^\sg$ for the posterior corresponding to signal $\sg$ and $\al_\sg$ for $\pr(\sg).$

Let $\val:\simp_M\mapsto\R$ be the principal's objective function.
For the bulk of our results, we consider the objective function $\val(\mu)$ for $\mu \in \simp_M$ that evaluates to the
row-player's payoff at equilibrium in the zero-sum game specified by $\A^\mu$.
Note that $\val(\mu)$ is unique, $\val(\mu):=\max_{x \in\simp_r}\min_{j\in[c]}(x^T \A^\mu)_j$,
  although there could be multiple Nash equilibria.

The quality of a signaling scheme $\al$ for a Bayesian zero-sum game is then given by
$\sum_{\mu\in\simp_M}\al_\mu \val(\mu)$. The {\em signaling problem} in a Bayesian zero-sum game is to find a signaling scheme $\al$
that maximizes $\sum_{\mu \in \simp_M} \al_\mu \val(\mu)$. Let $\opt(\I)$ denote the value of the optimal signaling scheme for a Bayesian zero-sum game $\I$. We note that $\opt(\I)$ is a concave function of the prior $\ld$, since if $\ld^1$ and $\ld^2$ form a convex decomposition of $\ld$, so do the optimal posteriors for $\ld^1$ and $\ld^2$.
By Caratheodory's theorem, $M+1$ posteriors (equivalently, signals) suffice to specify any convex decomposition of the prior. Together, this implies that an optimal signaling scheme can be specified by at most $M+1$ posteriors.

We say that an algorithm for the signaling problem is
an (additive) {\em $\ve$-approximation algorithm} if for every instance
$\I$
the algorithm runs in polytime and returns a signaling scheme of value at least
$\opt(\I)-\ve$.
A {\em polytime approximation scheme} (PTAS) is an algorithm that runs in polytime and
returns a solution of value at least $\opt(\I)-\ve$ for every instance $\I$ and constant
$\ve>0$;
an FPTAS is a PTAS whose running time for an instance $\I$ and parameter $\ve$ is
$\poly\bigl(\text{size of $\I$},\frac{1}{\ve}\bigr)$.

\paragraph{Security games}
Some of our results utilize a class of zero-sum games that we call
{\em extended security games},
wherein the payoff matrix for state $\tht$ is given by
\begin{equation}
\A^\theta := \bA + b^\theta \mone_c^T + \mone_r (d^\theta)^{T},
\qquad \qquad \text{where} \qquad b^\tht\in\R^r, d^\tht\in\R^c. \label{extsec1}
\end{equation}
Let ${B}$ and $D$ be matrices having columns $\{b^1,\ldots,b^M\}$, and
$\{d^1,\ldots,d^M\}$ respectively.
We obtain the following expressions for $\A^\mu$ and $\val(\mu)$ for $\mu\in \simp_M$.
\begin{equation}
\hspace*{-2ex}
\A^\mu\! = \bA + (B\mu)\mone_c^T + \mone_r (\mu^T {D}^T),
\ \ \val(\mu) = \max_{x \in \simp_r} \Bigl\{ x^T {B} \mu +
\min_{j\in[c]}\bigl(x^T \bA +  \mu^TD^T\bigr)_j \Bigr \}. \label{extsec}
\end{equation}

A special case of an extended security game (and the reason for this terminology) is
the {\em network security game} defined by~\cite{Dughmi14}.
Given an undirected graph $G=(V,E)$ with $n = |V|$ and a parameter
$\rho \ge 0$, the states of nature correspond to the vertices of the graph.
The row and column players are called attacker and defender respectively.
The attacker and defender's pure strategies correspond to nodes of $G$.
Let $B$ be the adjacency matrix of $G$, and set
$\bA = D^T = -\rho I_{n \times n}$.
Then, for a given state of nature $\tht\in V$, and pure strategies $a,d\in V$ of the
attacker and defender, the payoff of the attacker is given by
$e_a^T Be_\tht - \rho(e_a^T+e_\tht^T)e_d$.
The interpretation is that the attacker gets a payoff of 1 if he selects a vertex
$a$ that is adjacent to $\tht$.
This payoff is reduced by $\rho$ if the defender's vertex $d$ lies in $\{\tht,a\}$,
and by $2 \rho$ if $d=\tht=a$.

\paragraph{Planted clique and planted clique cover}
Some of our hardness results are based on the hardness of the {\em planted-clique} and
{\em planted clique cover} problems.
The latter problem was introduced by Dughmi \cite{Dughmi14}.

\begin{definition}[Planted clique cover problem $\mathbf{PCover}(n,p,k,r)$~\cite{Dughmi14}]
\label{def:pcover}
Let $G \sim \G(n,p,k,r)$ be a random graph
generated by:
(1) including every edge independently with probability $p$; and
(2)
for $i = 1, \ldots, r$,
picking a set $S_i$ of $k$ vertices uniformly at random,
adding all edges having both endpoints in $S_i$.
We call the $S_i$s the planted cliques and $p$ the background density.
We seek to recover a constant fraction of the planted cliques $S_1, \ldots, S_r$,
given $G \sim \G(n,p,k,r)$.
\end{definition}

In the \emph{planted clique problem} $\mathbf{PClique}(n,p,k)$, there is a single planted clique ($r=1$) and the goal is to recover this clique. The following hardness assumption for the
planted-clique problem
{has been used in deriving various hardness results.}

\begin{conjecture}[Planted-clique conjecture] \label{pcliquehard}
For some $k = k(n)$ satisfying $k=\omega(\log n)$ and $k = o(\sqrt{n})$,
  there is no probabilistic polytime algorithm that solves
  $\pclique\bigl(n,\frac{1}{2},k\bigr)$
  with constant success probability.
\end{conjecture}

\paragraph{The ellipsoid method.}
We utilize the ellipsoid method to translate hardness and approximation
\mbox{results for the dual of the signaling problem to signaling.}

\begin{theorem}[Chapters 4, 6 in~\cite{ellipsoidbook}; Section 9.2
    in~\cite{NemirovskiY83}] \label{ellipthm}
Let $X\sse\R^n$ be a polytope described by constraints having encoding length at most
$L$. Suppose that for each $y\in\R^n$, we can determine in time
$\poly(\text{size of $y$},L)$ if $y\notin X$ and if so, return
a hyperplane of encoding length at most $L$ separating $y$ from $X$.

\begin{enumerate}[(i), topsep=0ex, itemsep=0ex]
\item
The ellipsoid method can find a point $x\in X$ or determine that $X=\es$ in time
$\poly(n,L)$.

\item Let $h:\R^n\mapsto\R$ be a concave function and $K=\sup_{x\in X}h(x)-\inf_{x\in X}h(x)$.
Suppose we have a {\em value oracle} for $h$ that for every $x\in X$, returns
$\psi(x)$ satisfying $|\psi(x)-h(x)|\leq\dt$.
There exists a polynomial $p(n)$ such that for any $\e\geq p(n)\dt$,
we can use the {\em shallow-cut ellipsoid method} to find $x^*\in X$ such that
$h(x^*)\geq\max_{x\in X}h(x)-2\e$ (or determine $X=\es$) in time
$T=\poly\bigl(n,L,\log(\frac{K}{\e})\bigr)$ and using at most $T$ queries to the value
oracle for $h$.
\end{enumerate}
\end{theorem}

\section{The dual signaling problem} \label{dsp}
The signaling problem can be formulated as the following mathematical program,
\begin{equation}
\max \ \ \sum_{\mu\in\simp_M}\al_\mu\val(\mu) \quad \text{s.t.} \quad
\sum_{\mu\in\simp_M}\al_\mu \mu_\tht=\ld_\tht \ \ \frall \tht\in\Tht, \quad \al\geq 0.
\tag{P} \label{primal}
\end{equation}
Notice that any feasible $\al$ must also satisfy $\sum_{\mu \in \simp_M} \al_\mu = 1$; hence,
$\al$ is indeed a distribution over $\simp_M$, and a feasible solution to \eqref{primal}
yields a signaling scheme. Let $\opt(\ld)$ denote the optimal value of \eqref{primal},
  and note that this is a {\em concave} function of $\ld$.
Although \eqref{primal} has a linear objective and linear constraints, it is not quite a
linear program (LP) since there are an infinite number of variables.
Ignoring this issue for now, we consider
the following dual of \eqref{primal}.
\begin{equation}
\min \quad w^T\ld \qquad \text{s.t.} \qquad
w^T\mu \geq \val(\mu) \ \ \frall \mu\in\simp_M, \quad w\in\R^M. \tag{D} \label{dual}
\end{equation}
The separation problem for \eqref{dual} motivates the following
{\em dual signaling problem}.

\begin{definition}[Dual signaling with precision parameter $\ve$]
Given a Bayesian zero-sum game $\bigl(\Tht,\{\A^\theta\}_{\theta \in \Theta}, \lambda\bigr)$,
$w\in \mathbb{R}^M$, and $\ve>0$, distinguish between:
\begin{enumerate}[(i), nosep]
\item
$\val(\mu) \ge w^T\mu + \ve$ for some $\mu \in \simp_M$; if so return $\mu\in\simp_M$
s.t. $\val(\mu)\geq w^T\mu-\ve$;
\item
$\val(\mu) < w^T\mu - \ve$ for all $\mu \in \simp_M$.
\end{enumerate}
\noindent
The threshold signaling problem is the special case of dual signaling
where $w=\eta\mone_M$ for some $\eta\in\R$.
\end{definition}

Notice that the dual signaling problem is unconstrained:
$\ld$ plays no role.

\section{Bayesian zero-sum games} \label{zerosum}
We now prove the following results for signaling in Bayesian zero-sum games. We
show that the signaling problem does not admit an FPTAS unless {\em P}=\np
(Theorem~\ref{nofptasthm}) and does not admit a PTAS assuming the hardness of the
planted-clique problem (Theorem~\ref{noptasthm}).
Complementing these hardness results,
we present a PTAS for a structured class of extended security games (Theorem~\ref{ptas}).

\subsection{\nphard{}ness of obtaining an FPTAS} \label{nofptas}

\begin{theorem}[Corollary of Theorems~\ref{opt2sep} and~\ref{dsphard}] \label{nofptasthm}
There is no FPTAS for the signaling problem, even for network security games,
unless P=NP.
\end{theorem}

\begin{theorem} \label{lem:opttosep} \label{opt2sep}
There is a polynomial $q(M)$ such that an $\frac{\ve}{q(M)}$-approximation algorithm $\B$ for the
signaling problem $\I$ yields a polytime algorithm for the threshold signaling problem
$(\I,\eta\mone_M,\ve)$.
Thus, an FPTAS for the signaling problem yields an FPTAS for the threshold signaling
problem.
\end{theorem}

\begin{proof}
Let $\bigl(\Tht,\{\A^\theta\},\lambda\bigr), \eta\mone_M, \ve$ be the input to the
threshold signaling problem, with precision parameter $\ve$. Note that for any
$\mu\in\simp_M$, we have $-1\leq\opt(\mu)\leq 1$ since $|\A^\tht_{i,j}|\leq 1$ for all
$\tht, i, j$.
Let $p(M)$ be the polynomial given by part (ii) of Theorem~\ref{ellipthm}.
Set $q(M)=p(M)+1$.
We utilize part (ii) of Theorem~\ref{ellipthm} with $X=\simp_M$, $\dt=\frac{\ve}{q(M)}$,
$h(\cdot)=\opt(\cdot)$ (which is concave, as noted earlier),
$K=2$, and using $\B$ as the
imperfect value oracle, to find $z\in\simp_M$ in polytime such that
$\opt(z)\geq\max_{\mu\in\simp_M}\opt(\mu)-p(M)\dt$.
We run $\B$ on the prior $z$ to obtain a signaling scheme $\al$ of value
$v\geq\opt(z)-\ve$. If $v\geq\eta$, then we return that we
are in case (i) and one of the points $\mu\in\simp_M$ with $\al_\mu>0$ must satisfy
$\val(\mu)\geq\eta$.
If $v<\eta$, then we have
$\max_{\mu\in\simp_M}\val(\mu)\leq\max_{\mu\in\simp_M}\opt(\mu)<\eta+\bigl(p(M)+1)\dt$, so we
return that we are in case (ii).
\end{proof}

\begin{theorem} \label{lem:dualhardness} \label{dsphard}
There is no FPTAS for the threshold signaling problem, even for network
security games, unless P=NP.
\end{theorem}

\begin{proof}
The proof follows readily via a reduction from the
{\em balanced complete bipartite subgraph} (BCBS) problem~\cite{GareyJ79}, which illustrates the
convenience of working with the dual signaling problem.
In BCBS, given a bipartite graph $G=(V \cup W, E)$ and an integer $r\geq 0$, we
want to determine if $G$ contains $K_{r,r}$ (i.e., an $r\times r$ biclique).
Given a BCBS instance, set $\ve=\frac{1}{2n^8}$, where $n=|V|+|W|$, and
$\eta=1-(2n+1)\ve$.
We create a Bayesian network security game by letting $G$ be the graph in the network
security game, and setting $\rho = 2r n \ve$.
Recall that this means that states of nature correspond to nodes of $G$, so
$\Tht=V\cup W$, and the payoff matrix for a distribution $\mu\in\simp_\Tht$ is given by
\eqref{extsec} where $B$ is the adjacency matrix of $G$ and $\bA=D^T=-\rho I_{n\times n}$.
This creates an instance of the threshold signaling problem with precision parameter $\ve$;
the prior $\ld$ is irrelevant.
We show that solving this instance would decide the
BCBS-instance.

If $G$ has the required subgraph $V'$, $W'$, set $\mu_v = 1/r$ for all
$v \in V'$ and $x_v = 1/r$ for all $v \in W'$.
Then, by \eqref{extsec}, we have
$\val(\mu) \ge x^TB\mu-\rho\|\mu+x\|_\infty\geq 1-\rho/r=\eta+\ve$.
where we have $x^T B\mu = 1$ since $V'$, $W'$ form a complete bipartite subgraph.

Suppose there exists $\mu \in \simp_M$ so that $\val(\mu) \ge \eta - \ve$. We show then that
$G$ contains $K_{r,r}$.
Let $x$ be the equilibrium strategy of the attacker, so
$\val(\mu)=x^TB\mu-\rho\|\mu+x\|_\infty$.
Let $V' := \{v\in V\cup W: \mu_v \ge 1/n^3\}$ and $W' := \{v\in V\cup W: x_v \ge 1/n^3\}$.
Then $\sum_{v \in V'} \mu_v = 1- \sum_{v \not \in V'} \mu_v > 1- 1/n^2$.  Similarly
$\sum_{v\in W'}x_v>1-1/n^2$.
Every vertex in $V'$ must be adjacent to every vertex in $W'$, otherwise
$x^T B \mu \le 1 - 1/n^6<\eta$. Thus, $V'$ and $W'$ must be in different partitions. Assume
$V' \subseteq V$ and $W' \subseteq W$.
For each vertex $v$, $\mu_v + x_v \le\frac{(1+ 1/n)}{r}$, otherwise
$\val(\mu)<1-(2n+2)\ve$.
Hence,
$|V'| \ge \frac{\sum_{v \in V'} \mu_v}{(1 + 1/n)/r}> r \frac{1-1/n^2}{(1+1/n)} = r(1-1/n)$,
and therefore $|V'| \ge r$. Similarly $|W'| \ge r$, and this yields the $r\times r$
biclique.
\end{proof}

We conjecture that Theorem~\ref{opt2sep} can, in fact, be strengthened to show that an
$\ve$-approximation for signaling yields an $O(\ve)$-approximation for threshold
signaling, so that a PTAS for signaling yields a PTAS for threshold signaling.
This would {\em rule out a sub-quasipolytime approximation scheme} (i.e.,
an $n^{\tilde\Om(\log^{1-o(1)} n)}$-time approximation scheme) for signaling
under the (deterministic) {\em exponential time hypothesis} (ETH),
since we prove in Section~\ref{extn} that there is no sub-quasi-PTAS
for threshold signaling assuming ETH.

This would be an {\em optimal} hardness result since a quasi-PTAS follows from~\cite{CCDEHT15}.
Recently, Rubinstein~\cite{rubinstein15} obtains this hardness result via a \emph{direct} reduction
that builds upon ideas in \cite{AIM14}.
However, tightening part (ii) of Theorem~\ref{ellipthm}
  would give a much simpler proof.
We leave this as an intriguing open question.
Below, we rule out a PTAS for
signaling under an orthogonal hardness assumption.

\subsection{Planted-clique hardness of obtaining a PTAS} \label{noptas} \label{sec:ptas}

\begin{theorem}
\label{noptasthm}
There is a constant $\ve_0$ such that, assuming the planted-clique hardness conjecture
  (Conjecture~\ref{pcliquehard}), there is no $\ve_0$-approximation for the signaling
  problem in Bayesian zero-sum games.
\end{theorem}

Our hardness result strengthens the one in~\cite{Dughmi14}, which rules out an
FPTAS assuming the planted-clique conjecture.
The reduction therein creates a network security game from a graph
$G\sim\G\bigl(n,\frac{1}{2},k,r\bigr)$ (see Section~\ref{prelim}).
The idea is that if a signaling scheme achieves value close to 1, then it must place
a large weight on posteriors and attacker mixed-strategies that randomize over a large
set of nodes.
Further, the posterior and attacker must essentially identify
dense components of $G$, as otherwise the attacker's value would be close to the background
density $\frac{1}{2}$.
As noted earlier, a limitation of this type of construction
is that
the parameter $\rho$ used in the network security game needs to be
roughly $\Omega(\log n)$
to ensure that the posterior and the attacker's mixed strategies are supported on
an $\Omega(\log n)$-size set of nodes.
This only yields an $\Tht\bigl(\frac{1}{\polylog(n)}\bigr)$ gap, which is insufficient
to rule out a PTAS.
We overcome this obstacle by moving away from a network security game, and instead
exploiting an idea of~\cite{FNS07}
  to eliminate all equilibria of $O(\log n)$-size support from the game.
Theorem \ref{noptasthm} follows immediately by combining Lemmas~\ref{lem:family_t}
and~\ref{lem:main}.

\begin{lemma}
\label{lem:family_t}
Let $\epsilon > 0$,
  $k = k(n)$ satisfy $k = \omega(\log n)$ and $k = o(\sqrt{n})$, and $r = \Theta(n/k)$.
Suppose there is a polytime algorithm that takes as input $G \sim \G\bigl(n,\frac{1}{2},k,r\bigr)$
with planted cliques $\set{S_i}$,
  and outputs a family $\T\sse 2^V$ of clusters satisfying the following with constant
  probability, for any constant $c_3\geq 10^3:$

\vspace{-1.25ex}
\begin{equation*}
\text{for an $\epsilon$-fraction of $\set{S_i}$,
$\exists T \in \T$ with $|T \cap S_i|\geq\max\bigl\{\e|T|, c_3 \log n\bigr\}$.}
\tag{*} \label{pcover}
\end{equation*}

\noindent
Then there is a polynomial-time algorithm for $\mathbf{PClique}\bigl(n,\frac{1}{2},k\bigr)$
having constant success probability.
\end{lemma}

\begin{lemma} \label{lem:main}
Let $k = k(n)$ satisfy $k = \omega(\log n)$ and $k = o(\sqrt{n})$,
  and $r = \frac{5n}{k}$.
There is a polynomial-time randomized reduction that takes a graph
  $G \sim \G\bigl(n,\frac{1}{2},k,r\bigr)$ as input
  and outputs a Bayesian zero-sum game
  such that the following hold with high probability.
\begin{itemize}[topsep=0.25ex, itemsep=0ex, leftmargin=0ex]
\item[] (Completeness)\quad There is a signaling scheme having value at least $0.99$.
\item[] (Soundness)\quad Given a signaling scheme of value at least $0.97$,
  one can obtain in polytime a collection $\T$ of clusters satisfying condition
  \eqref{pcover} in Lemma~\ref{lem:family_t}.
\end{itemize}
\end{lemma}

Above, and throughout this section, when we say with high probability, we mean success
probability $1-\frac{1}{\poly(n)}$.
The Bayesian zero-sum game we construct
always admits a signaling scheme of large value;
  however {\em finding} a near-optimal signaling scheme in polytime
  would refute the planted-clique conjecture.
Lemma~\ref{lem:family_t} (proved in Appendix \ref{append-family_t}) is similar to a
planted-clique recovery result proved in~\cite{Dughmi14}.
While we utilize similar ideas,
our
result works under {\em much weaker} requirements. Our lemma
allows clusters in $\T$ to have size $\Tht(\log n)$ --- which is crucial for
Lemma~\ref{lem:main} --- whereas in~\cite{Dughmi14}, the clusters need to have size
$\omega(\log^2 n)$.
In the rest of this section, we prove Lemma~\ref{lem:main}.
  We use the following parameters.
\begin{equation}
  Z = 20, \quad c_2 = 10^5, \quad
  c_1 = c_2 \log(4Z/3) + 2, \quad N = n^{c_1}.
\end{equation}
To keep the presentation simple, we give a construction where
  $\A^\theta_{i,j}\in[-Z, Z]$ (as opposed to $[-1,1]$).
Let $A_G$ denote the ($n\times n$) adjacency matrix of $G=(V,E)$.
We split $G$ into $G^-$ and $G^+$ with corresponding adjacency matrices $A_G^-$ and $A_G^+$
  where $G^-$ are the background edges and $G^+$ are the clique edges added in
  steps (1) and (2) of Definition \ref{def:pcover} respectively.
The states of nature and the row-player's strategies correspond to the nodes of $G$.
The prior $\ld$ is $\mone_n/n$, thus each state of nature (each vertex) is equally likely to occur.
For every $\tht\in\Tht=V$, the payoff matrix $\A^\theta \in [-Z,Z]^{n \times (2N+1)}$
is given by $[a^\tht\ \ B\ \ \mone_n (d^\tht)^T]$, which are defined as follows:

\begin{enumerate}[(1), topsep=0.5ex, itemsep=0ex, labelwidth=\widthof{(3)}, leftmargin=!]

\item $a^\theta$ is
the $\theta$-th column of the adjacency matrix $A_G$, so $a^\tht_i=1$ if $(i,\tht)\in E$
and is 0 otherwise.
\item $B$ is an $n \times N$ matrix, where each $B_{i,j}$ is set independently to $2-Z$
with probability $\frac{3}{4Z}$, and $2$ otherwise.
\item $d^\tht\in[-Z,Z]^N$, where each entry $d^\theta_j$ is set independently to
$2-Z$ with probability $\frac{3}{4Z}$, and $2$ otherwise.
\end{enumerate}

We use $\Row$ and $\Col$ to denote the row and column players respectively.
Let $D$ be the $n\times N$ matrix having rows $(d^\tht)^T$ for $\tht\in\Tht$.

To gain some intuition, observe that for a
posterior $\mu$ and $\Row$'s mixed strategy $x$,
the row vector $x^T\A^\mu$ yielding $\Col$'s payoffs is $[x^TA_G\mu\ \ x^TB\ \ \mu^TD]$.
Thus, if $\Col$ plays action $1$ (with probability 1),
  the expected payoff of $\Row$ is equal to $x^TA_G\mu$.
If $\mu$ and $x$ are uniform over $S, T \subseteq V$,
  the expected payoff is exactly
\[ \text{bi-density}_G(S, T):=
\frac{\left|\{(u,v)\in S\times T: \{u,v\}\in E\}\right|}{|S||T|}. \]
The remaining $2N$ pure strategies of $\Col$
  are used to force the principal and $\Row$ to choose a posterior $\mu$ and
  mixed strategy $x$ respectively that are ``well spread out''.

The average of the entries in any column of $B$ or $D$ is $\frac{5}{4}>\max_i a^\tht_i$.
Exploiting this, Claim~\ref{largesmallsup}(i)
implies that if $x$ and $\mu$ both randomize uniformly over a large set of vertices,
$\Col$ plays column 1.
The completeness proof now follows from the oft-used idea of
(roughly speaking) choosing
posteriors
and mixed strategies for $\Row$
that randomize uniformly over the planted cliques.
Conversely, if $x$ or $\mu$ has support of size at most $c_2\log n$, then
Claim~\ref{largesmallsup}(ii) implies that $\Col$ can play some column of $B$ or $D$ and make $\val(\mu)$ negative
Thus, in order to obtain value close to 1, both $\mu$ and $\Row$ have to randomize over
$\Omega(\log n)$-size sets of nodes. Using this, one can carefully extract a collection of
node-sets satisfying condition \eqref{pcover} of Lemma~\ref{lem:family_t}.
This yields the soundness proof.

The following properties about the above construction will be useful.

\begin{claim} \label{largesmallsup}
Let $R\sse V$.
(i) If $|R|=\w(\log n)$, with high probability, for every $j\in [N]$,
$\frac{1}{|R|}{\sum_{i\in R}B_{i,j}}>1$ and $\frac{1}{|R|}\sum_{i\in R}D_{i,j}>1$.
(ii) If $|R|\leq c_2\log n$, with high probability,
$\exists j,k\in[N]$ such that $B_{i,j}=2-Z=D_{i,k}$ for all $i\in R$.
\end{claim}

\begin{proof}
We first prove (i). The proof is a standard application of Chernoff bounds, and is also essentially shown
in~\cite{hazankrauthgamer}.
We prove the statement for $B$;
the argument for $D$ is identical.
Fix a column $j\in [N]$. We have
$\Ex\bigl[\frac{\sum_{i\in R}B_{i,j}}{|R|}\bigr]=\frac{5}{4}$, where the expectation is
over the random construction of $B$. Since $|R|=\w(\log n)$, the size of $R$ is large
enough so that  Chernoff bounds imply that
$\Pr\bigl[\frac{\sum_{i\in R}B_{i,j}}{|R|}<\frac{9}{8}\bigr]\leq \frac{1}{2N\poly(n)}$.
The union bound over all $N$ columns yields the claim.

We now prove (ii). The proof again follows from Chernoff bounds, and is the key insight in~\cite{FNS07} (also
utilized in~\cite{hazankrauthgamer}).
Fix some $R\sse V$ with $|R|=c_2\log n$. We prove the statement for $B$; the proof for $D$
is identical.
For a given $j\in[N]$, we have
$\Pr[\exists i\in R\text{ s.t. }B_{i,j}\neq 2-Z]=1-\bigl(\frac{3}{4Z}\bigr)^{|R|}$.
So
\[
\Pr[\forall j\in[N], \exists i\in R\text{ s.t. }B_{i,j}\neq 2-Z]
=\bigl[1-\bigl(\frac{3}{4Z}\bigr)^{|R|}\bigr]^N.
\]
Taking the union bound over all $R\sse V$ with $|R|=c_2\log n$,
we obtain
\begin{align*}
\Pr\Bigl[\exists R\sse V&\text{ with $|R|=c_2\log n$ s.t. no $j\in[N]$ satisfies
    $B_{i,j}=2-Z$ for all $i\in R$}\Bigr] \\
& \leq {n\choose{c_2\log n}}\biggl[1-\Bigl(\tfrac{3}{4Z}\Bigr)^{|R|}\biggr]^N
\leq \exp\Bigl({c_2\log^2 n-N\bigl(\tfrac{3}{4Z}\bigr)^{c_2\log n}}\Bigr) \\
& \leq 1-1/\poly(n). \qedhere
\end{align*}
\end{proof}

\begin{lemma}[Proposition B.2 in~\cite{Dughmi14} quantified] \label{dug:bidens}
Let $\ve>0$, and $c\geq 24 \cdot 2.1 \cdot\max\bigl\{1,\frac{1+\ve}{\ve^2}\bigr\}$.
For all $n\geq 2$, we have

\[ \Pr\bigl[\exists S, T\sse V\text{ with }|S|,|T|\geq c\log n,\
\bidens_{G^-}(S,T)>\frac{1+\ve}{2}\bigr]\leq\frac{2}{n^3} \]
\end{lemma}

\begin{lemma}[Corollary of Lemma~\ref{dug:bidens}] \label{lem:bidensity}
For $c_2 = 10^5$ and $\eps = 0.03$
With high probability, for all $S, T\sse V$ with $|S|, |T| \ge c_2 \log n$,
$\bidens_{G^-}(S,T)\leq\frac{1+\e}{2}$.
\end{lemma}

\subsubsection{Completeness proof in Lemma~\ref{lem:main}}
We use a deterministic signaling scheme that groups together states of nature in the same
planted clique. Let $S_1, \ldots, S_r$ be the planted cliques in $G$ in some arbitrary
order.
Let $S_i' = S_i \setminus \bigcup_{1 \le j < i} S_j$ for $i\in[r]$ be the set of vertices in $S_i$ that do not appear in earlier cliques. Define $A:=V\sm\bigcup_j S_j$ as the remaining vertices. Finally, $S_0' = A \cup \bigl\{ v \in S_i' : |S_i'| < \frac{k}{10^4}\bigr\}$. Our signaling scheme is $(\Sg,\alpha,\mu)$ where the set of signals is $\Sg=\{0\}\cup\bigl\{i\in[r]: |S'_i|\geq\frac{k}{10^4}\bigr\}$. For each signal $\sg$, $\alpha_\sg =\frac{|S_\sg'|}{n}$ and $\mu_\sg$ is the uniform distribution over $S_\sg'$. Note that the signaling scheme is independent of $B$ and $D$.

For posterior $\mu^\sg$, where $\sg\neq 0$, consider the strategy $x^\sg$ where $\Row$
plays the uniform distribution on $S_\sg'$.
Claim~\ref{largesmallsup}(i) implies that $\Col$'s best response to $x^\sg$ is
to play column 1.
Therefore,
$\val(\mu^\sg)\geq\bidens(S'_\sg,S'_\sg)=1 -\frac{1}{|S_\sg'|}\geq 1-\frac{10^4}{k}$.
With $r = \frac{5n}{k}$, we have
$|A|\leq e^{0.1}\cdot\Ex[|A|]\leq e^{-4.9}n$ with high probability due to
standard Chernoff bounds (since the events $\{v\in A\}_{v\in V}$ are negatively
correlated).
Therefore, for suitably large $n$, with high probability,
$|S'_0|\leq |A|+\frac{5n}{k}\cdot\frac{k}{10^4}\leq e^{-4.7}n$.
So, with high probability, the signaling scheme has value at least
{$\sum_{\sg\in\Sg\cap[r]} \alpha_\sg\bigl(1-\frac{10^4}{k}\bigr)
\ge(1 - e^{-4.7})\bigl(1 - \frac{10^4}{k}\bigr)\ge 0.99$.}

\subsubsection{Soundness proof in Lemma~\ref{lem:main}}
For a signal $\sg\in\Sg$ with corresponding posterior $\mu_\sg$,
let $x_\sg$ denote $\Row$'s equilibrium strategy for $\A^{\mu_\sg}$.
We first filter out the set of ``useful'' signals, i.e., those with relatively high value.
Let $\Sg_1=\{\sg\in\Sg: \val(\mu_\sg)\geq 1-\sqrt{\e}\}$.
We show that for all $\sg\in\Sg_1$, $\mu_\sg$ and $x_\sg$ place a significant mass
over a large set of nodes, and use this insight to extract clusters.
Fix $\e=0.03$.
For every signal $\sg\in\Sg_1$,
define $T_\sigma = \set{i: e_i^T A_G\mu_\sigma \ge 1 - \frac{Z\sqrt{\e}}{Z-2}}$,
and let $\tilde{x}_\sigma$ be the uniform distribution on $T_\sigma$.
We output $\T = \set{T_\sigma : \sigma \in \Sigma_1}$.

We show that $\T$ satisfies condition \eqref{pcover} in Lemma~\ref{lem:family_t}.
The value of the signaling scheme is
  $\sum_{\sigma \in \Sigma} \alpha_\sigma \val(\mu_\sigma) \ge 1 - \epsilon$.
Noting that $\val(\mu)\leq 1$ for all $\mu$, by Markov's inequality, we have
$\al(\Sigma_1) \ge 1 - \sqrt{\epsilon}$. (Given a vector $v\in\R^k$, and $S\sse[k]$, we
use $v(S)$ to denote $\sum_{i\in S}v_i$.)
Assume that the high probability event in Claim~\ref{largesmallsup}(ii) happens.

Fix $\sg\in\Sg_1$.
For any $R\sse V$ with $|R|\leq c_2\log n$, we must have
$x_\sg(R)\leq\frac{2}{Z}$ and $\mu_\sg(R)\leq\frac{2}{Z}$. Otherwise, suppose
$x_\sg(R)>\frac{2}{Z}$ (the argument for $\mu_\sg$ is similar). Then, considering the
column $j$ of $B$ having $B_{i,j}=2-Z$ for all $i\in R$, we have
$\sum_{i\in[n]}(x_\sg)_iB_{i,j}\leq (2-Z)x_\sg(R)+2\bigl(1-x_\sg(R)\bigr)<0$, which
implies that $\val(\mu_\sg)<0$.
Now since $1-\sqrt{\e}\leq\val(\mu_\sg)\leq 1$, by the definition of $T_\sg$ and Markov's
inequality, we have $x_\sg(T_\sg)\geq\frac{2}{Z}$, and hence
$|T_\sigma| \ge c_2 \log n$.
We now switch from $x_\sg$ to $\tx_\sg$ in order to relate the value of the signaling
scheme to bi-density and deduce that $\T$ satisfies condition \eqref{pcover}.
As before, $G^-$ are the background edges and $G^+$ are the clique edges added in
  steps (1) and (2) of Definition \ref{def:pcover} respectively,
  and $A_G^-$ and $A_G^+$ are the corresponding adjacency matrices.
  Let $A_G^i$ be the adjacency matrix of the clique $S_i$.
Note that $A_G\leq A_G^-+A^+_G\leq A^-_G+\sum_{i=1}^rA_G^i$.

Let $R$ denote the $c_2 \log n$ largest entries in $\tilde{x}_\sigma^TA_G^-$,
  and let $\tilde{\mu}_\sigma$ be the uniform distribution on $R$.
Since $\tilde{\mu}_\sigma$ and $\tilde{x}_\sigma$ are uniform distributions over $R$ and
$T_\sg$ respectively (which have size at least $c_2\log n$), we have
$\tx_\sg^TA_G^-\tmu_\sg=\bidens(T_\sg,R)\leq\frac{1+\e}{2}$ due to
Lemma~\ref{lem:bidensity}.
Moreover,
$\mu_\sigma(R) \le \frac{1}{10}$,
  and since the maximum entry of $\tilde{x}_\sigma^T A_G^-$
  outside of $R$ is at most the average entry in $R$, we have
$\tilde{x}_\sigma^T A_G^- \mu_\sigma
  \le \frac{1}{10} + \frac{9}{10} \cdot \tilde{x}_\sigma^T A_G^- \tilde \mu_\sigma<0.6$.

Finally, we also have
  $\sum_{\sigma \in \Sigma_1} \alpha_\sigma (\tilde{x}_\sigma^TA_G\mu_\sigma)
  \ge (1 - \sqrt{\epsilon})\bigl(1 - \tfrac{Z\sqrt{\e}}{Z-2}\bigr) > 0.85$.
Therefore,
\begin{equation*}
\begin{split}
\frac{1}{4}
    &< \sum_{\sigma \in \Sigma_1} \alpha_\sigma\tilde{x}_\sigma^T (A_G - A_G^-) \mu_\sigma
    \le \sum_{\sigma \in \Sigma_1} \alpha_\sigma \sum_{i=1}^r \tilde{x}_\sigma^T A_G^i\mu_\sigma \\
& = \sum_{i=1}^r \sum_{\sigma \in \Sigma_1} \alpha_\sigma \mu_\sigma(S_i) \frac{|T_\sigma \cap S_i|}{|T_\sigma|}
\le \sum_{i=1}^r \Bigl( \sum_{\sigma \in \Sigma_1} \alpha_\sigma \mu_\sigma(S_i) \Bigr)
\left( \max_{T \in \T} \frac{|T \cap S_i|}{|T|} \right) \\
& \overset{(**)}{\le}
\sum_{i=1}^r \frac{|S_i|}{n} \left( \max_{T \in \T} \frac{|T \cap S_i|}{|T|} \right)
    = \frac{5}{r} \sum_{i=1}^r \left( \max_{T \in \T} \frac{|T \cap S_i|}{|T|} \right).
\end{split}
\end{equation*}
Inequality $(**)$
follows since for every $v\in\Tht$, we have
$\sum_{\sg\in\Sg_1}\al_\sg(\mu_\sg)_v$ is at most
$\sum_{\sg\in\Sg}\al_\sg(\mu_\sg)_v=\ld_v=\frac{1}{n}$.
Therefore
$\frac{1}{r} \sum_{i=1}^r \left( \max_{T \in \T} \frac{|T \cap S_i|}{|T|} \right) \ge \frac{1}{20}$.
This implies that at least a $\frac{1}{39}$-fraction of $S_1, \ldots, S_r$
satisfy $\max_{T\in\T}\frac{|T\cap S_i|}{|T|}\geq\frac{1}{40}$.
Since $|T| \ge c_2 \log n$ for all $T\in\T$, $\T$ satisfies condition \eqref{pcover} in Lemma~\ref{lem:family_t}.

\subsection{A PTAS for structured extended security games}
\label{ptasextsec} \label{sec:algos}
We now devise a PTAS for a structured class of extended security games (Theorem~\ref{ptas}).
First, we reduce the signaling problem to the dual
signaling problem using the ellipsoid method (Theorem~\ref{sep2opt}). This reduction applies to
{\em all} Bayesian zero-sum games. Next, we devise a PTAS for the dual signaling problem
for our class of extended network security games (Theorem~\ref{dspptas}).

\begin{theorem}[Dual signaling to signaling] \label{lem:septoopt} \label{sep2opt}
A polytime algorithm for the dual signaling problem with precision $\ve$ gives a
$5\ve$-approximation algorithm for the signaling problem. In particular, a PTAS for
the dual signaling problem yields a PTAS for the signaling problem.
\end{theorem}

To prove Theorem~\ref{sep2opt},
we utilize the ellipsoid method, specifically, part (i) of Theorem~\ref{ellipthm},
adapting the standard transformation from separation to optimization to take into account
the additive error in the dual separation problem.
To circumvent the technical difficulties caused by the infinite-dimensionality of
\eqref{primal}, we approximate \eqref{primal} by a finite-dimensional LP, where we
restrict the variables in \eqref{primal}, and, analogously the constraints in \eqref{dual}
to a suitable $\dt$-net of $\simp_M$.
Let $\I=\bigl(\Tht,\{\A^\tht\},\ld\bigr)$ be a Bayesian zero-sum game. Recall that
$|\A^\tht_{i,j}|\leq 1$ for all $\tht, i, j$.
For $\delta \in (0,1]$ with $1/\delta \in \mathbb{Z}$, and $\mu \in \simp_M$, define
\[
S_\delta:=\bigl\{\mu' \in \simp_M: \, \mu_\theta' / \delta \in \mathbb{Z} \quad \forall \theta\in \Theta\bigr\},
\quad
S_\delta(\mu) := \{\mu' \in S_\delta : \Vert \mu-\mu' \Vert_\infty \le \delta\}.
\]

\begin{claim} \label{perterr}
Fix $\mu\in\simp_M$. For any $\mu' \in S_\delta(\mu)$, we have
$\vert \val(\mu)-\val(\mu') \vert\le M \delta$. Hence, we can efficiently find
$\hmu\in S_\dt(\mu)$ such that $w^T\hmu-\val(\hmu)\leq w^T\mu-\val(\mu)+M\dt$.
\end{claim}

\begin{proof}
The entries in $\A^\mu$ and $\A^{\mu'}$ differ by at most $M\delta$. Hence for every
mixed-strategy profile $(x,y) \in \simp_r\times\simp_c$, we have
$|x^T (\A^\mu - \A^{\mu'}) y| \le M \delta$, and therefore
$\vert \val(\mu) - \val(\mu') \vert \le M \delta$.

We can efficiently find $\hmu\in S_\dt(\mu)$ that minimizes $w^T\mu'$ over $\mu'\in S_\dt(\mu)$
since this can be cast as an LP. Then, we have
$w^T\hmu-\val(\hmu)\leq w^T\mu-(\val(\mu)-M\dt)$.
\end{proof}

We work with the following finite-dimensional counterparts of \eqref{primal} and
\eqref{dual} and argue that this approximation only yields a small error.

\hspace*{-8ex}
\begin{minipage}[]{0.54\textwidth}
\begin{align*}
\begin{split}
\max & \quad \sum_{\mu\in S_\dt}\al_\mu\val(\mu) \\
\text{s.t.} & \quad \sum_{\mu\in S_\dt}\al_\mu \mu =\ld; \ \ \al \geq 0.
\end{split} \tag{P$_\dt$} \label{primd}
\end{align*}
\end{minipage}
\quad
\begin{minipage}[]{0.46\textwidth}
\begin{align*}
\begin{split}
\min & \quad w^T\ld  \\
\text{s.t.} & \quad w^T\mu \geq \val(\mu) \quad \forall \mu\in S_\dt.
\end{split} \tag{D$_\dt$} \label{duald}
\end{align*}
\end{minipage}

\medskip\noindent
Since $\lambda \in \conv(S_\delta(\lambda))$, \eqref{primd} is feasible for any
$\lambda\in \simp_M$.
Clearly, any solution to \eqref{primd} gives a solution to \eqref{primal} of equal
value. The converse is also approximately true.

\begin{lemma} \label{clm:finiteapprox}
Any feasible solution $\al$ to \eqref{primal} of value $v$ gives a solution
to \eqref{primd} of value at least $v - M \delta$.
Hence, $\opt\Pd\geq\opt\text{\eqref{primal}}-M\dt$.
\end{lemma}

\begin{proof}
This is an easy consequence of Claim~\ref{perterr}.
For any $\mu \in S$, let $\tau^{(\mu)} \in \simp_{S_\delta(\mu)}$ be some convex decomposition of
$\mu$. Then
\[
\lambda_\theta = \sum_{\mu \in S} \al_\mu \mu_\theta
= \sum_{\mu \in S} \al_\mu \sum_{\mu' \in  S_\delta(\mu)} \tau_{\mu'}^{(\mu)} \mu'_\theta
= \sum_{\mu' \in S_\delta} \mu'_\theta  \sum_{\mu \in S} \al_\mu \tau_{\mu'}^{(\mu)} \,.
\]

\noindent Thus, setting $\al'_{\mu'} := \sum_{\mu \in S} \al_\mu \tau_{\mu'}^{(\mu)}$ for all
$\mu' \in S_\delta$, we obtain that $\al'$ is a feasible solution to \eqref{primd}.
To compare the objective values of $\al$ and $\al'$, note that
\begin{align*}
\sum_{\mu \in S} \al_\mu \val(\mu) & = \sum_\mu \al_\mu\val(\mu) \sum_{\mu' \in S_\delta(\mu)} \tau_{\mu'}^{(\mu)}
\le \sum_\mu \al_\mu \sum_{\mu' \in S_\delta(\mu)} \tau_{\mu'}^{(\mu)}\bigl(\val(\mu') + M \delta\bigr)  \\
& = \sum_{\mu' \in S_\delta} \al'_{\mu'} \val(\mu') + M \delta \, . \qedhere
\end{align*}
\end{proof}

Now the basic idea is to solve \eqref{duald} with the ellipsoid method using the
algorithm $\B$ to obtain a separation oracle for \eqref{duald} with an additive error. In
the course of solving \eqref{duald}, we also obtain a polynomial-size LP
consisting of the violated inequalities of \eqref{duald} returned by the separation oracle
during the execution of the ellipsoid method whose optimal value is the same as $\opt\Dd$.
Taking the dual of this compact LP yields an LP of the same form as \eqref{primd} but
with $\al_\mu$ variables for only polynomially many points in $S_\dt$; solving this yields
the desired approximate signaling scheme.
The additive error in the separation oracle for \eqref{duald} complicates the arguments
slightly.

We now discuss the details.
Set $\dt=\ve/M$.
Let $\B$ be the algorithm for solving the dual signaling problem with precision $\ve$.
For a given $\nu,\e\in\R$, consider the set
$Q(\nu,\e):=\{w\in\R^M: w^T\ld\leq\nu, \quad w^T\mu\geq\val(\mu)-\e \ \ \forall \mu\in S_\dt\}$.
Note that the constraints of $Q(\nu,\e)$ have encoding length
$\poly\bigl(M,\text{size of $(\ld,\nu,\dt,\e)$}\bigr)$.
For a given $\nu$ and $w\in\R^M$, we can determine if $w\in Q(\nu,\ve)$, or find a
hyperplane separating $w$ from $Q(\nu,-2\ve)$, as follows.
We first check if $w^T\ld\leq\nu$ and if not, then return this as the separating
hyperplane. We run $\B$ on the input $(\I,w,\ve)$.
If $\B$ determines that we are case (i), then it also returns $\mu\in\simp_M$ with
$\val(\mu) \ge w^T\mu - \ve$. By Claim~\ref{perterr}, we can then find $\hmu\in S_\dt(\mu)$ such
that $w^T\hmu-\val(\hmu)\leq 2\ve$, so we can use $w^T\hmu-\val(\hmu)\geq 2\ve$ to separate
$w$ from $Q(\nu,-2\ve)$. If $\B$ determines that we are in case (ii), then we are
certainly not in case (i), so we have $\val(\mu)\leq w^T\mu+\ve$ for all $\mu\in\simp_M$, which
implies that $w\in Q(\nu,\ve)$.

So for a fixed $\nu$, in polynomial time, the ellipsoid method either certifies that
$Q(\nu,-2\ve)=\es$ or returns a point in $Q(\nu,\ve)$. We find the smallest $\nu$ (via
binary search) such that the latter case happens; call this value $\nu^*$. Then,
\[
\nu^*\geq
\Bigl(\min\ \ w^T\ld \quad \text{s.t.} \quad w^T\mu\geq\val(\mu)-\ve \ \ \forall \mu\in S_\dt\Bigr)
=\opt\Pd-\ve.
\]
The equality above follows since the dual of the minimization LP is above is \eqref{primd}
with the objective function changed to
$\sum_{\mu\in S_\dt}\al_\mu\bigl(\val(\mu)-\ve\bigr)=\sum_{\mu\in S_\dt}\al_\mu\val(\mu)-\ve$.
For any $\e>0$, running the ellipsoid method for $\nu=\nu^*-\e$ yields a polynomial-size
certificate for the emptiness of $Q(\nu^*-\e,-2\ve)$ consisting of the inequality
$w^T\ld\leq\nu^*-\e$ and the polynomially many violated inequalities
$w^T\mu-\val(\mu)\geq 2\ve$ returned during the execution of the ellipsoid method.
Let $T\sse S_\dt$ be the polynomial-size set of points for which we obtain these violated
inequalities. By duality,
\begin{equation*}
\begin{split}
\nu^*-\e
& <\Bigl(\min w^T\ld \ \ \text{s.t.} \ \ w^T\mu\geq\val(\mu)+2\ve \ \forall \mu\in T\Bigr) \\
& =2\ve+
\Bigl(\max\ \sum_{\mu\in T}\al_\mu\val(\mu)\ \ \text{s.t.} \ \ \sum_{\mu\in T}\mu\al_\mu=\ld,
\ \ \al\geq 0\Bigr).
\end{split}
\end{equation*}
Thus, solving the polynomial-size LP inside the parentheses yields a signaling scheme of value at
least $\opt\Pd-3\ve-\e$, so taking $\e=\ve$ and using Lemma~\ref{clm:finiteapprox}, we obtain a signaling scheme of value at
least $\opt\text{\eqref{primal}}-5\ve$. \qed

\begin{definition}[\boldmath $\gm$-Lipschitz] \label{def:regular}
A matrix $A \in \mathbb{R}^{r \times c}$ is
{\em $\gm$-Lipschitz} if
{$\|x^TA - x'^TA\|_\infty\leq\gm\|x-x'\|_\infty$} for all $x$, $x' \in \simp_r$.
An extended security game specified by matrices $\bA, B, D$ (see \eqref{extsec1},
\eqref{extsec}) is \mbox{$\gm$-Lipschitz} if ${D}^T$ is $\gamma$-Lipschitz.
We place no constraints on the matrices $\bA$ and ${B}$.
\end{definition}

Observe that an extended security game specified by matrices $\bA, B, D$ (see \eqref{extsec1},
\eqref{extsec}) is $\gm$-Lipschitz if ${D}^T$ is $\gamma$-Lipschitz.
We place no constraints on the matrices $\bA$ and ${B}$.
We design a simple PTAS for the dual signaling problem on $\gm$-Lipschitz extended
security games, for constant $\gm$.
By Theorem~\ref{lem:septoopt}, this yields a PTAS for the
signaling problem for $\gm$-Lipschitz extended security games.

\begin{theorem} \label{dspptas} \label{ptas}
There is a PTAS for the dual signaling problem on $\gm$-Lipschitz
extended security games.
This yields a PTAS for the signaling problem on $\gm$-Lipschitz extended-security games.
\end{theorem}

\begin{proof}
Given Theorem~\ref{sep2opt}, we only need to prove the first statement.
Let $(\I,w,\ve)$ be the input to the dual signaling problem where $\I$ is a
$\gm$-Lipschitz extended security game.
Set $\ve' = \ve/\gamma$.
Our algorithm simply
finds $\hmu=\argmax_{\mu \in S_{\ve'}}\bigl(\val(\mu) - w^T\mu\bigr)$ by exhaustive search.
If $\val(\hat{\mu})-w^T\hat{\mu}\geq 0$, we state that we are in case (i) and return
$\hat{\mu}$; else we state that we are in case (ii).

First, note that the algorithm runs in time
$\poly\bigl(\text{size of $\I$},M^{\frac{\gm}{\ve}}\bigr)$, since
$|S_{\ve'}|\le{M\choose{1/\ve'}}\bigl(\frac{1}{\ve'}\bigr)^{1/\ve'}$ (there are
${M\choose{1/\ve'}}$ choices for the support, and at most $\frac{1}{\ve'}$ choices for
each of the at most $\frac{1}{\ve'}$ coordinates in the support).

Let $\mu^*$ maximize $\val(\mu) - w^T\mu$, and $x^*$ be the equilibrium strategy
for the row player in the resulting zero-sum game.
We claim that $\val(\mu^*) -w^T\mu^* \le \val(\hat{\mu}) - w^T \hat{\mu} + \ve$, which
shows that we correctly solve the dual signaling problem: if case (i) applies, then
$\val(\hmu)-w^T\hmu\geq 0$; if case (ii) applies, then clearly, $\val(\hmu)-w^T\hmu<-\ve$.

We now prove the claim. Since $\mu^* \in \conv(S_{\ve'}(\mu^*))$, there exists some
$\mu' \in S_{\ve'}(\mu^*)$ such that
$x^{*T} {B} \mu^* - w^T \mu^*  \le x^{*T} {B} \mu' - w^T \mu'$.
Further, since ${D}^T$ is $\gamma$-Lipschitz, for all $j \in [c]$,
$\bigl(x^{*T} \bA + \mu^{*T} {D}^T \bigr)_j \le \bigl(x^{*T} \bA + \mu'^{T} {D}^T \bigr)_j + \gamma \ve'$.
Combining these inequalities yields that $\val(\mu')-w^T\mu'\geq \val(\mu^*)-w^T\mu^*-\ve$.
\end{proof}

\section{Bayesian network routing games} \label{sec:congestion}
We now consider the signaling problem in Bayesian network routing games and prove an
{\em optimal} inapproximability result for linear latency functions:
It is \nphard to obtain a multiplicative approximation better than $4/3$
(Theorem~\ref{thm:routinghardness}), and this approximation is achieved for linear latency
functions
by a simple signaling scheme that simply reveals the state of nature
(Theorem~\ref{fullrev}).

A {\em network routing game}
is a tuple
$\Gamma =\bigl(G=(V,E),\{l_e\}_{e\in E},\{(s_i, t_i,d_i)\}_{i\in[k]}\bigr)$, where $G$ is
a directed graph
with latency function $l_e:\R_+\mapsto\R_+$ on each edge $e$.
Each $(s_i, t_i,d_i)$ denotes a {\em commodity};
$d_i$ specifies the volume of flow routed from
$s_i$ to $t_i$ by self-interested agents, each of whom controls an infinitesimal amount of
flow and selects an $s_i$-$t_i$ path as her strategy.
A strategy profile thus corresponds to a multicommodity flow composed of $s_i$-$t_i$
flows of volume $d_i$ for all $i$; we call any such flow a feasible flow.
The latency on edge $e$ due to a flow $f$ is given by $l_e(f_e)$, where $f_e$ is the total
flow on $e$. The latency of a path $P$ is $l_P(f):=\sum_{e\in P}l_e(f_e)$.
The total latency of a flow $f$
is $C(l; f):=\sum_{e\in E}f_el_e(f_e)$;
an optimal flow is a feasible flow with minimum latency.
A feasible flow $f$ in a routing game is a {\em Nash flow} (also called
a {\em Wardrop flow}), if each player chooses a minimum latency path;
that is,
for all $i$, all $s_i$-$t_i$ paths $P$, $Q$ with $f_e>0$ for all $e\in P$, $l_P(f)\leq l_Q(f)$.
All Nash flows have the same total latency (see, e.g., \cite{RT02}).

In a {\em Bayesian network routing game}, the edge latency functions
$\{l_e^\theta\}_{e \in E}$ may depend on the state of nature $\theta\in\Tht$ (and, as
before, we have a prior $\ld\in\simp_\Tht$).
The principal seeks to {\em minimize} the latency of the Nash flow.
Given $\mu \in \simp_\Tht$,
the expected latency function on each edge $e$ is
$l_e^\mu(x_e) := \sum_{\theta \in \Theta} \mu_\theta l_e^\theta(x_e)$.
Define $\val(\mu):=C(l^\mu; f^\mu)$, where $f^\mu$ is the Nash flow for latency functions
$\{l^\mu_e\}$.
The {\em signaling problem in a Bayesian routing game} is to determine
$(\al_\mu)_{\mu\in\simp_M}\geq 0$ of finite support specifying a convex decomposition
of $\lambda$ (i.e., $\sum_{\mu\in\simp_M}\al_\mu\mu=\ld$) that minimizes
the expected latency of the Nash flow, $\sum_{\mu\in\simp_M} \al_\mu\val(\mu)$.

\smallskip
\begin{theorem}
For any $\epsilon > 0$, obtaining a $(4/3-\epsilon)$-approximation for the signaling
problem in Bayesian routing games is NP-hard, even in single-commodity games with linear
latency functions.
\label{thm:routinghardness}
\end{theorem}

Let $\bigl(G,s,t,d\bigr)$ be a single-commodity routing game.
We reduce from the problem of determining edge tolls $\tau\in\R_+^E$ that minimize
$C\bigl(l+\tau; f^\NE(\tau)\bigr)$, where $l+\tau$ denotes the collection of latency
functions $\{l_e(x)+\tau_e\}_e$ and $f^\NE(\tau)$ is the Nash flow for $l+\tau$.
Note that $C(l+\tau; f)=\sum_e f_e(l_e(f_e)+\tau_e)$ takes into account the contribution
from tolls; we refer to this as the total {\em cost} of $f$.
By {\em optimal tolls}, we mean tolls $\tau$ that minimize
$C\bigl(l+\tau; f^\NE(\tau)\bigr)$.

\begin{theorem}[\cite{ColeDR06}]
There are optimal tolls where the toll on every edge is $0$ or $\infty$.
If  P $\neq$ NP, there is no
$\left(\frac{4}{3}-\epsilon\right)$-approximation algorithm for the problem of computing
optimal tolls in networks with linear latency functions, for any $\epsilon > 0$.
\label{thm:tolls}
\end{theorem}

Let $\Gamma =\bigl(G=(V,E),l,s,t,d\bigr)$ be an instance of a routing game with linear
latencies. Let $m=|E|\geq 2$. By scaling latency functions suitably, we may assume that
$d=1$. Then, for any latency functions $l'$, the latency of the Nash flow for $l'$ equals
the common delay of all flow-carrying $s$-$t$ paths.
Let $L=C(l; f^\NE)$ be the latency of the Nash flow for $l$.
Let $\tau^*$ be optimal $\{0,\infty\}$-tolls,
$L^*=C\bigl(l+\tau^*, f^\NE(\tau^*)\bigr)$ be the optimal cost,
and $K^*:=\{e \in E: \tau_e^* =\infty\}$.
We can view $\tau^*$ as simulating the removal of edges in $K^*$.

We create the following Bayesian routing game.
Let $\bigl(G_1=(V_1, E_1),s_1,t_1\bigr)$ and $\bigl(G_2=(V_2, E_2),s_2,t_2\bigr)$
be two copies of $(G,l,s,t)$.
Add vertices $s$, $t$, and edges $(s,s_1)$, $(s,s_2)$ and
$(t_1,t)$, $(t_2,t)$. Call the graph thus created $H$.
For $e\in E_1\cup E_2$ with corresponding edge $e'\in E$, set the latency function in the new graph $h_e(x)=l_{e'}(x)$, and
set $h_e(x)=0$ for $e=(s,s_1),(s,s_2),(t_1,t),(t_2,t)$.
The states of nature correspond to edges in $H$.
We set $\ld_\tht=1/m^2$ for all $\tht\in E_1\cup E_2$;
the remaining $1-\frac{2}{m}$ mass is spread equally on $(s,s_1)$, $(s,s_2)$.
We set $h_e^\tht(x)=h_e(x)+8m^3L$ if $\tht=e$ and $h_e(x)$ otherwise.
Our Bayesian routing game is $\bigl((G,\{h^\theta_e\}_{\tht,e},s,t,d),\lambda\bigr)$.

The idea here is that state $\tht$ encodes the removal of edge $\tht$: specifically,
if $\mu_\tht=\Omega\bigl(\frac{1}{m}\bigr)$ for a posterior $\mu$, then $h^\mu$ simulates
removing edge $\tht$ due to the large constant term $8m^3L$.
Let $K_i$ be the edge-set corresponding to $K^*$ in $G_i$, for $i=1,2$.
The prior $\ld$ is set up so that: (a) it admits a convex
decomposition into posteriors $\mu^1,\mu^2$, where $h^{\mu^i}$ simulates that $G_i\sm K_i$
is connected to $s$ and $G_{3-i}$ is disconnected from $s$; and
(b) any convex-decomposition of $\ld$ must be such that a large weight is placed on
posteriors $\mu$, where $h^\mu$ simulates that only one of $G_i$ is connected to $s$, so
that $\{\mu_e8m^3L\}_{e\in E_i}$ yields tolls $\tau$ for edges in $E$ such that
$C\bigl(l+\tau, f^\NE(\tau)\bigr)\leq\val(\mu)$.
Lemma~\ref{lem:dualhardnessmult}
makes the statements in (a) and (b) precise, and Theorem~\ref{thm:routinghardness} follows
immediately from Lemma~\ref{lem:dualhardnessmult} and Theorem~\ref{thm:tolls}.

\begin{lemma}
There is a signaling scheme for
the above Bayesian routing game with latency $L^*$.
Further, given a signaling scheme $\al$ for the above Bayesian routing game with expected latency
$L'$, one can obtain tolls $\tau$ such that the routing game $(G,l+\tau,s,t,d)$ has Nash
latency at most $\frac{L'}{1-4/m}$.
\label{lem:dualhardnessmult} \label{sig2tolls}
\end{lemma}

\begin{proof}
We first show the existence of a signaling scheme with latency $L^*$.
Define posterior $\mu^1\in\simp_{E_H}$ as:
$\mu^1_\tht=2/m^2$ for all $\tht\in K_1 \cup E_2 \setminus K_2$,
$\mu^1_{(s,s_2)}=(1-2/m)$. Define $\mu^2$ symmetrically as:
$\mu^2_\tht=2/m^2$ for all $\tht\in K_2 \cup E_1 \setminus K_1$,
$\mu^2_{(s,s_1)}=(1-2/m)$.
Then $\lambda =(\mu^1 + \mu^2) /2$, and this is our signaling scheme. We will show that
$\val(\mu^1) = \val(\mu^2) \le L^*$, proving the lemma.

Consider distribution $\mu^1$; the argument for $\mu^2$ is symmetrical.
The idea is that an edge $e$ with $\mu_e^1 > 0$ has $h^{\mu^1}_e(x)\geq 8mL$, which
effectively deletes $e$ from $H$; other edges have $h^{\mu_1}_e(x)=h_e(x)$.
So $\mu^1$ simulates retaining edges in $G_1\sm K_1$.
Let $f=f^\NE(\tau^*)$ be the Nash flow in the routing game $(G,l+\tau^*,s,t,d)$. So
$C(l+\tau^*; f)=L^*$. Recall that $d=1$, so every $s$-$t$ path in $G$ has latency at least
$L^*$. Then the flow that sends $d$ on edges
$(s,s_1)$ and $(t_1,t)$ and $f$ on edges of $G_1$, is feasible.
On every edge $e\in E(H)$ with positive flow, $\mu_e^1 = 0$, so the latency of this flow
under $h^{\mu^1}$ is $L^*$.
Further, this is a Nash flow for $h^{\mu^1}$:
any $s$-$t$ path $P$ either contains an edge with $\mu_e^1 > 0$, and if not,
contains an $s_1$-$t_1$ path; in the latter case, there is a corresponding $s$-$t$ path
$Q$ in $G$, and the latency of $P$ under $h^{\mu^1}$ equals $(l+\tau^*)_Q(f)$, which is
at least $L^*$ since $f$ is the Nash flow for $l+\tau^*$.

Next, we show how to obtain the required tolls from the signaling scheme $\al$ (with
expected latency $L'$).
Assume $L'\leq L$, otherwise $\tau=0$ suffices.
At least $(1-4/m)$ of the probability mass of $\al$ must be on posteriors $\mu$ with
$\mu_{(s,s_1)} +\mu_{(s,s_2)} \ge 1/m$.
There must exist such a posterior $\mu'$ with $\val(\mu')\le\frac{L'}{1-4/m}$.
Assume $\mu'_{(s,s_1)}\geq\frac{1}{2m}$; the other case
is symmetric.
Let $f=f^{\mu'}$ be the Nash flow for latency functions $h^{\mu'}$.
(Again, since $d=1$, every $s$-$t$ path $P$ in $H$ with $f_e>0$ for all $e\in P$ satisfies
$h^{\mu'}_P(f)=\val(\mu')$.)

Since $h^{\mu'}_{(s,s_1)}\geq 4m^2L>\val(\mu')$, we must have $f_{(s,s_1)}=0$,
so $f$ is supported on $G_2$.
Abusing notation, for $e\in E_2$, we also use $e$ to denote the corresponding edge
in $E$. For every $e \in E_2$,
we have $h^{\mu'}_e(x)=l_{e}(x) + \mu_e' 8 m^3 L$.
Thus, defining $\tau_e = \mu_e' 8 m^3 L$ for all $e \in E$, we obtain
that $f$ restricted to $E_2$ is a Nash flow for $(G,l+\tau,s,t,d)$, and its latency is
at most $\val(\mu')$.
This is easy to see, since every $s$-$t$ path $P$ in $G$ corresponds to an $s_2$-$t_2$
path $Q$ in $H$, and $(l+\tau)_P(f)=h^{\mu'}_Q(f)$.
\end{proof}

\begin{theorem} \label{fullrev}
The full-revelation signaling scheme, i.e., revealing the state of nature, has the price
of anarchy for the underlying latency functions as its approximation
ratio. In particular, for linear latencies, it achieves a $\frac{4}{3}$-approximation.
\end{theorem}

\begin{proof}
Recall that the price of anarchy (\poa) for a class of latency functions
is the maximum ratio, over all instances involving these latency functions,
of the latencies of the Nash flow and optimal flow.
For linear latency functions, the \poa is $\frac{4}{3}$~\cite{RT02}.

Intuitively, the result follows because full-revelation is the best signaling
scheme if one seeks to minimize the expected latency of the {\em optimal} flow, and the
multiplicative error that results from this change in objective (from the latency of the
Nash flow to that of the optimal flow) cannot exceed the price of anarchy.

Slightly abusing notation, we use $f^\tht$ to denote the Nash flow with respect to the
latency functions $\{l_e^\tht\}$. We use $\brf^\tht$ to denote the optimal flow for
latency functions $\{l_e^\tht\}$.
Let $\rho$ be the price of anarchy for the collection $\{l^\tht_e\}_{e\in E,\tht\in\Tht}$
of latency functions, so we have
$C(l^\tht; \brf^\tht)\geq C(l^\tht; f^\tht)/\rho$ for all $\tht\in\Tht$.
The full-revelation signaling scheme has cost
$\sum_{\tht\in\Tht}\ld_\tht C(l^\tht; f^\tht)$.

Consider any signaling scheme $\al$. Its cost is
\begin{align*}
\sum_{\mu\in\simp_M}\al_\mu\val(\mu) &=\sum_{\mu}\al_\mu C(l^\mu; f^\mu)
=\sum_{\mu}\al_\mu\sum_{\tht\in\Tht}\mu_{\tht} C(l^\tht; f^\mu)
\geq \sum_{\mu,\tht}\al_\mu\mu_{\tht} C(l^\tht; \brf^\tht) \\
& \geq \sum_{\mu,\tht}\al_\mu\mu_{\tht} C(l^\tht; f^\tht)/\rho
=\sum_\tht \ld_\tht C(l^\tht; f^\tht)/\rho. \qedhere
\end{align*}
\end{proof}

\section{Extensions: hardness results for related problems} \label{extn}

\subsection{Maximum prior problem}
We study the closely-related problem
of finding $\mu\in\simp_M$ that maximizes $\opt(\mu)$.
The proof of Theorem~\ref{opt2sep} in fact shows that
{\em a PTAS for the maximum-prior problem yields a PTAS for threshold signaling}.
Theorem~\ref{maxposthm} uses this implication to rule out a PTAS for the
maximum prior problem under the exponential time hypothesis (ETH) by giving a simple, clean reduction
from the best-Nash problem in {\em general-sum} two-player games,
for which a PTAS is ruled out by~\cite{BKW15}.
Theorem~\ref{maxposthm} establishes the {\em optimal} hardness result for the
maximum prior problem, since
a quasi-PTAS for the maximum prior problem
was recently presented in \cite{CCDEHT15}.

Theorem~\ref{maxposthm} also implies that the general maximum prior problem studied in \cite{CCDEHT15}
  does not have a PTAS under the ETH, when the objective function is $O(1)$-Lipschitz but not $O(1)$-noise-stable.
  (For this case, \cite{CCDEHT15} ruled out a PTAS assuming hardness of planted clique.)
  This is because the objective function (minimax value) for signaling in zero-sum games is $O(1)$-Lipschitz.

\begin{theorem} \label{maxposthm}
Assuming ETH, there is a constant $\ve_0$ such that, any algorithm that returns an
(additive) $\ve_0$-approximation for the maximum prior problem,
even for extended security games, must run in quasipolynomial, i.e.,
$n^{\tilde\Om(\log^{1-o(1)} n)}$, time. In particular, assuming ETH, there is no PTAS for the
maximum prior problem, even for extended security games.
\end{theorem}
Recall that, as noted earlier, the proof of Theorem~\ref{opt2sep} shows that, for any
$\ve$, a polytime $\ve$-approximation for the maximum prior problem yields a polytime
algorithm for the threshold signaling problem with precision parameter $2\ve$. Thus, it
suffices to show that, assuming ETH, there is some constant $\e_0$ such that solving the
threshold signaling problem with precision parameter $\e_0$, even for extended security
games, requires quasipolynomial running time. To show this, we reduce from the problem of
finding an $\e$-Nash equilibrium in a general two-player game with $\e$-approximate
social welfare, and utilize the following hardness result for this problem.

\begin{theorem}[\cite{BKW15}] \label{bestnash}
Assuming ETH, there is a constant $\epsilon^* > 0$ such that any algorithm for finding an
$\epsilon^*$-approximate Nash equilibrium with social welfare at least $\OPT -\epsilon^*$
in a general bimatrix game requires $n^{\tilde\Om(\log^{1-o(1)} n)}$ time,
where $\OPT$ is the optimal welfare of a Nash equilibrium.
\end{theorem}

\begin{proofof}{Theorem~\ref{maxposthm}} Let $(\Rc,\C)$ be a bimatrix game, where $\Rc,\C\in[-1,1]^{m\times n}$ are the payoffs for the
row- and column- players respectively. To avoid confusion with the extended security game,
we refer to the row- and column- players in the bimatrix game as the $\Rc$- and $\C$- players.
A pair of mixed strategies $(\x,\y)$ for the $\Rc$- and $\C$- players respectively is an
$\ve$-approximate equilibrium if:
\begin{equation}
\x^T(\Rc+\C)\y-\max_{i\in[m]} (\Rc\y)_i-\max_{j\in[n]} (\x^T\C)_j \ge-\ve. \label{apxnash}
\end{equation}
The social welfare of $(\x,\y)$ is defined as $\x^T(\Rc+\C)\y$.
Let $\OPT$ be the maximum social welfare of a (mixed) Nash equilibrium of $(\Rc,\C)$.
Note that $-2\leq\OPT\leq 2$.

We construct an extended security game where the states of nature correspond to the pure
strategies of the $\C$-player (in the bimatrix game), and the row-player's pure strategies
(in the extended security game) correspond to the $\Rc$-player's strategies (in the bimatrix
game).
We will set things up so that the expected payoff in the extended security
game to the row player under a posterior distribution $\mu$ and when he plays a mixed
strategy $x$ is a linear combination of the LHS of \eqref{apxnash} (viewing $(x,\mu)$ as a
mixed-strategy profile for the bimatrix game $(\Rc,\C)$) and the social welfare
of $(x,\mu)$ in the bimatrix game $(\Rc,\C)$.
Let $\e>0$ be a parameter.
The payoffs in the extended security game will have absolute value at most $1+O(1/\eps)$.
We will show that solving the threshold signaling problem for the resulting extended
security game with threshold $\eta=\eta'-\e$,
and precision parameter $\eps$ yields
a $6\eps$-approximate Nash equilibrium of $(\Rc,\C)$ with social welfare at least
$\eta'-2\eps$, whenever there is a Nash equilibrium of $(\Rc,\C)$ with social welfare at
least $\eta'$ or we state that we are in case (i) of the threshold signaling problem.
So via binary search, we can obtain a $6\e$-approximate Nash equilibrium of $(\Rc,\C)$
with social welfare at least $\OPT-3\e$.
Thus, setting $\e_0=\Tht(\e^{*^2})$,%
\footnote{The $\Tht(e^{*^2})$ is because
we need additive error $\Tht(\e^*)$ when payoffs
are bounded in absolute value by $1+O(1/\e^*)$; when we scale payoffs so that they lie in
$[-1,1]$, this translates to an $\Tht(\e^{*^2})$-approximation.}
where $\e^*$ is as given by
Theorem~\ref{bestnash}, we obtain that, assuming ETH, the threshold signaling problem with
precision parameter $\e_0$ requires quasipolynomial time, completing the proof.

We proceed to describe the extended security game and prove the desired claim.
We set $\Tht=[n]$, so $M=n$. The row-player's pure strategy set is $[m]$, and the
column-player's pure-strategy set is $[m]\times[n]$, so the row- and column- players have
$r=m$ and $c=mn$ pure strategies respectively.
The $r\times c$ matrix $\bA$, $r\times M$ matrix $B$, and $c\times M$ matrix $D$ in the
extended security game are
\begin{alignat*}{2}
\bA_{i,(i',j)}\ & =\ -\frac{1}{\e}\C_{i,j} \qquad \qquad && \forall i\in[m], (i',j)\in[m]\times[n] \\
B_{i,j}\ & =\ \Bigl(1+\frac{1}{\e}\Bigr)(\Rc_{i,j}+\C_{i,j}) \qquad && \forall i\in[m], j\in[n] \\
D_{(i,j'),j}\ & =\ -\frac{1}{\e}\Rc_{i,j} && \forall (i,j')\in[m]\times[n], i\in[m].
\end{alignat*}

\begin{claim} \label{clm:minequal}
For all $\mu\in\simp_M$, $x \in \simp_r$, we have
\[
\min_{k\in[c]}\bigl(x^T\bA+\mu^TD^T\bigr)_j \quad = \quad
-\frac{1}{\e}\Bigl(\max_{i\in [m]}(\Rc \mu)_i+\max_{j\in[n]}(x^T\C)_j\Bigr).
\]
\end{claim}

\begin{proof}
Consider any column-player strategy $k=(i',j')\in[m]\times[n]$.
We have
\begin{equation*}
\begin{split}
(x^T\bA)_k+(\mu^TD^T)_k
& =\sum_{i\in[m]}x_{i}\bA_{i,(i',j')}+\sum_{j\in[n]}\mu_{j}D_{(i',j'),j} \\
& =-\frac{1}{\e}\Bigl(\sum_{i\in[m]}x_{i}\C_{i,j'}+\sum_{j\in[n]}\mu_{j}\Rc_{i',j}\Bigr)
=-\frac{1}{\e}\Bigl((x^T\C)_{j'}+(\Rc \mu)_{i'}\Bigr). \qedhere
\end{split}
\end{equation*}
\end{proof}

It follows from Claim~\ref{clm:minequal} that for any $\mu\in\simp_M$, we have
\begin{align}
\hspace*{-1ex}
\val(\mu) &=\max_{x\in\simp_r}\Bigl(x^TB\mu+\min_{j\in[c]}\bigl(x^T\bA+\mu^TD^T\bigr)_j\Bigr)
\notag \\
& =\max_{x\in\simp_m}\biggl[\Bigl(1+\frac{1}{\e}\Bigr)x^T(\Rc+\C)\mu
-\frac{1}{\e}\Bigl(\max_{i\in [m]}(\Rc \mu)_i+\max_{j\in[n]}(x^T\C)_j\Bigr)\biggr] \notag \\
& = \max_{x\in\simp_m}\biggl[x^T(\Rc+\C)\mu
-\frac{1}{\e}\Bigl(\max_{i\in [m]}(\Rc \mu)_i+\max_{j\in[n]}(x^T\C)_j-x^T(\Rc+\C)\mu\Bigr)\biggr]
\label{valcomp}
\end{align}
Now suppose we solve the threshold signaling problem with threshold $\eta=\eta'-\e$
(where $\eta'\leq 2$) and precision parameter $\e$. Suppose $(x^*,\mu^*)$ is a Nash
equilibrium of $(\Rc,\C)$ with social welfare at least $\eta'$. It follows from
\eqref{valcomp} that $\val(\mu^*)\geq\eta'$. So we are not in case (ii) of the threshold
signaling problem, and
must obtain $\mu\in\simp_n$ such that $\val(\mu)\geq\eta-\e=\eta'-2\e$. From \eqref{valcomp},
this implies that there is $x\in\simp_m$ such that $x^T(\Rc+\C)\mu\geq\eta'-2\e$ and
\[
\max_{i\in[m]}(\Rc \mu)_i+\max_{j\in[n]}(x^T\C)_j-x^T(\Rc+\C)\mu
\leq\e\Bigl(x^T(\Rc+\C)\mu-\eta'+2\e\Bigr)\leq 6\e
\]
where the last inequality follows since $\Rc,\C\in[-1,1]^{m\times n}$.
The same calculation holds whenever we state that we are in case (i) and return
$\mu\in\simp_n$.
\end{proofof}

\subsection{Hardness with other equilibrium notions}
It is known that in zero-sum games, correlated
equilibria and Nash equilibria are payoff-equivalent, that is, they yield the same payoffs
(this was also noted in~\cite{Dughmi14}).
Thus, our hardness results extend to the
case of correlated equilibria, as well as other notions of stability that are
payoff-equivalent to Nash equilibria in zero-sum games. To see this,
note that for $\mu \in \simp_M$, due to the payoff equivalence, $\val(\mu)$ is also the payoff
of the row player in any correlated equilibrium in the zero-sum game specified by
$\A^\mu$. Hence, the statement of the signaling problem and its optimal value remain
unchanged. Further, any signaling scheme for correlated equilibria gives a signaling
scheme for Nash equilibria of equal value. This immediately extends {\em all} our hardness
results (Theorem~\ref{nofptasthm}, Theorem~\ref{dsphard}, Theorem~\ref{noptasthm},
Theorem~\ref{maxposthm}) to correlated equilibria (and other payoff-equivalent
equilibria).

\subsection{Signaling with general objective functions}
We now consider a more general signaling problem in Bayesian zero-sum games, where the
principal's value
may depend on the players' strategies, and show that it is \nphard to obtain a PTAS.

\noindent Formally, we have a Bayesian zero-sum game
  $\bigl(\Tht,\{\A^\theta\}_{\theta \in \Theta},\lambda\bigr)$
and a $\Tht\times r\times c$ {\em principal objective tensor}
$\F=\bigl(\F^\tht(i,j)\bigr)$; that is, $\F^\theta \in [-1, 1]^{r \times c}$ for all
$\tht\in\Tht$.
  We now define $\val(\mu) = \max_{(x_\mu, y_\mu) \in \mathit{NE}(\A^\mu)}
  x_\mu^T(\sum_\tht\mu_\tht \F^{\tht}) y_\mu$,
  where $\mathit{NE}(\A^\mu)$ is the set of all (exact) Nash equilibria of $\A^\mu$.
  As before, we seek a signaling scheme $(\Sg,\al,\mu)$ that maximizes
  $\sum_{\sigma \in \Sigma} \al_\sigma \val(\mu_\sigma)$.
  
\begin{theorem} \label{stackthm}
Given a Bayesian zero-sum game $\bigl(\Tht,\{\A^\theta\}_{\theta \in \Theta},\lambda\bigr)$,
  and a principal objective tensor $\F$,
  it is NP-hard to distinguish whether the optimal signaling scheme has value $0$
  or at least $\frac{1}{2}$.
\end{theorem}
The \nphard{}ness proof follows from a reduction
  from the {\em balanced vertex cover} (BVC) problem proposed in~\cite{conitzer06}.
In BVC, we are given a graph $G = (V, E)$,
  and we want to know if $G$ has a vertex cover of size $\frac{|V|}{2}$.
Given an instance of BVC with $n$ nodes, we construct the following Bayesian zero-sum game
where the states of nature correspond to nodes of $G$ and the prior is $\lambda = \mone_n / n$.

  The row player's pure strategy is to pick a node $v_1 \in V$,
  and the column player's pure strategy is to either pick a vertex $v$, an edge $e$,
  or a special strategy $s$.
We design the {\em column player's} payoff as follows. The payoff for strategy
\[
v\ \text{is}\ \left\{ \begin{array}{ll}
   \frac{n}{n-2} & \mbox{if $v \notin \{\tht, v_1\}$,} \\
    0 & \mbox{otherwise.} \end{array} \right.
\qquad
e\ \text{is}\ \left\{ \begin{array}{ll}
    \frac{n}{n-2} & \mbox{if $e$ is not incident with $\tht$,} \\
    0 & \mbox{otherwise.} \end{array} \right.
\qquad
s\ \text{is}\ 1.
\]

The principal's objective tensor is set up so that
  he is interested only in getting the column player to play the strategy $s$,
  that is,
  $\F^\theta(v,s)=1$ for all $\tht,v\in V$; all other entries of $\F$ are 0.

\begin{lemma}
The Bayesian zero-sum game defined above has a signaling scheme of value at least
$\frac{1}{2}$ if and only if $G$ has a vertex cover of size $\frac{n}{2}$.
\end{lemma}

\begin{proof}
First, suppose $G$ has a vertex cover $C$ with $|C| = \frac{n}{2}$.
The principal simply signals if $\theta\in C$ or not. That is, $\ld$ is decomposed as
$(\mu^1+\mu^2)/2$, where $\mu^1_v=\frac{2}{n}$ for all $v\in C$ (and $0$ otherwise), and
$\mu^2_v=\frac{2}{n}$ for all $v\notin C$.
For posterior $\mu^1$, there is a Nash equilibrium
  where the row player chooses the mixed strategy $x$ that picks $v_1 \in V \setminus C$
  uniformly at random and the column player chooses strategy $s$; thus, the principal gets
  a value of 1.
This is because every node and edge is ``protected'' with probability at least
$\frac{2}{n}$;
the payoff of the column player for a pure strategy $v$ or $e$ is therefore at most
$\frac{n}{n-2}\bigl(1-\frac{2}{n}\bigr)\leq 1$.
Since $\mu^1$ is chosen with probability $\frac{1}{2}$, this signaling scheme achieves
value at least $\frac{1}{2}$.

On the other hand, we show that if $\mu$ is a posterior with $\val(\mu)>0$, then $G$ has a BVC
solution.
Let $(x,y)$ be a Nash equilibrium that attains value $\val(\mu)$, that is,
$\val(\mu)=x^T(\sum_\tht\mu_\tht\F^\tht)y$.
Since $\val(\mu)>0$, we must have $y_s>0$.
For this to happen,
  every node in $V$ must be protected with probability at least $\frac{2}{n}$.
That is, we must have $\frac{n}{n-2}(1-x_v)(1-\mu_v)\leq 1$ for all $v\in V$.
Then, $n-2+\sum_vx_v\mu_v=\sum_v(1-x_v)(1-\mu_v)\leq n-2$, which implies that we must have
$x_v\mu_v=0$ and $(1-x_v)(1-\mu_v)=1-\frac{2}{n}$ for all $v\in V$.
So it must be that for all $v \in V$, exactly one of $\mu_v$ and $x_v$ is equal to
$\frac{2}{n}$.
Let $C = \set{v: \mu_v > 0}$. It follows that $|C| = \frac{n}{2}$.
The payoff of a column player for an edge $e=(u,v)$ is $\frac{n}{n-2}(1-\mu_u-\mu_v)$,
which must be at most $1$, so we have $\mu_u+\mu_v\geq\frac{2}{n}$.
It follows that $C$ is a vertex cover of $G$.
\end{proof}

We remark that it is important to allow the principal's payoff to depend on specific
strategies,
  and also to enforce exact Nash equilibrium.
Intuitively, these two ingredients together make the objective function
  $\val(\mu)$ very ``sensitive'' in $\mu$.
Moreover, these two conditions are essentially necessary for an \nphard{}ness result,
  as Cheng et al.~\cite{CCDEHT15} gave a bi-criteria quasi-PTAS for this general signaling problem,
  i.e., a quasi-polytime algorithm that loses an additive $\epsilon$
  in the objective as well as in the Nash equilibrium constraints.

\subsubsection*{Acknowledgments.}
The authors are grateful to Shaddin Dughmi for
various suggestions, including on the planted clique reduction
and for suggesting the network routing games problem.
We also thank David Kempe, and Li Han for helpful discussions.

\bibliographystyle{alpha}
\bibliography{agt,signal,signaling}

\newcommand{\etalchar}[1]{$^{#1}$}
\newcommand{\SortNoop}[1]{}
\begin{thebibliography}{DGGP11}

\bibitem[AIM14]{AIM14}
Scott Aaronson, Russell Impagliazzo, and Dana Moshkovitz.
\newblock {AM} with multiple {M}erlins.
\newblock In {\em {IEEE} 29th Conference on Computational Complexity, {CCC}
  2014}, pages 44--55, 2014.

\bibitem[Ake70]{akerloflemons}
George~A Akerlof.
\newblock The market for ``lemons'': Quality uncertainty and the market
  mechanism.
\newblock {\em The quarterly journal of economics}, pages 488--500, 1970.

\bibitem[AV14]{amesvavasis}
Brendan~PW Ames and Stephen~A Vavasis.
\newblock Convex optimization for the planted k-disjoint-clique problem.
\newblock {\em Mathematical Programming}, 143(1-2):299--337, 2014.

\bibitem[BBM13]{bergemanndiscrimination}
Dirk Bergemann, Benjamin Brooks, and Stephen Morris.
\newblock The limits of price discrimination.
\newblock {\em Economic Theory Center Working Paper}, (052-2013), 2013.

\bibitem[BKW15]{BKW15}
Mark Braverman, Young~Kun Ko, and Omri Weinstein.
\newblock Approximating the best nash equilibrium in $n^{o(\log n)}$-time
  breaks the exponential time hypothesis.
\newblock In {\em ACM-SIAM Symposium on Discrete Algorithms (SODA)}, 2015.

\bibitem[Bla51]{blackwell51}
David Blackwell.
\newblock Comparison of experiments.
\newblock In {\em Second Berkeley Symposium on Mathematical Statistics and
  Probability}, volume~1, pages 93--102, 1951.

\bibitem[BMS12]{miltersensignals}
Peter Bro~Miltersen and Or~Sheffet.
\newblock Send mixed signals: earn more, work less.
\newblock In {\em Proceedings of the 13th ACM Conference on Electronic Commerce
  (EC)}, pages 234--247, 2012.

\bibitem[CCD{\etalchar{+}}15]{CCDEHT15}
Yu~Cheng, Ho~Yee Cheung, Shaddin Dughmi, Ehsan Emamjomeh-Zadeh, Li~Han, and
  Shang-Hua Teng.
\newblock Mixture selection, mechanism design, and signaling.
\newblock In {\em 56th Annual Symposium on Foundations of Computer Science
  (FOCS)}, 2015.

\bibitem[CDR06]{ColeDR06}
Richard Cole, Yevgeniy Dodis, and Tim Roughgarden.
\newblock How much can taxes help selfish routing?
\newblock {\em J. Comput. Syst. Sci.}, 72(3):444--467, 2006.

\bibitem[CS06]{conitzer06}
Vincent Conitzer and Tuomas Sandholm.
\newblock Computing the optimal strategy to commit to.
\newblock In {\em Proceedings of the 7th ACM conference on Electronic Commerce
  (EC)}, 2006.

\bibitem[DGGP11]{dekel}
Yael Dekel, Ori Gurel-Gurevich, and Yuval Peres.
\newblock Finding hidden cliques in linear time with high probability.
\newblock In {\em ANALCO}, pages 67--75. SIAM, 2011.

\bibitem[DIR14]{DIR14}
Shaddin Dughmi, Nicole Immorlica, and Aaron Roth.
\newblock Constrained signaling in auction design.
\newblock In {\em the 25th ACM Symposium on Discrete Algorithms (SODA)}, 2014.

\bibitem[Doe11]{doerr11}
Benjamin Doerr.
\newblock Analyzing randomized search heuristics: tools from probability
  theory.
\newblock In {\em Theory of randomized search heuristics}. World Scientific,
  2011.

\bibitem[DP09]{dubhashi2009}
Devdatt~P Dubhashi and Alessandro Panconesi.
\newblock {\em Concentration of measure for the analysis of randomized
  algorithms}.
\newblock Cambridge University Press, 2009.

\bibitem[Dug14]{Dughmi14}
Shaddin Dughmi.
\newblock On the hardness of signaling.
\newblock In {\em 55th {IEEE} Annual Symposium on Foundations of Computer
  Science, {FOCS}}, pages 354--363, 2014.

\bibitem[EFG{\etalchar{+}}12]{emeksignaling}
Yuval Emek, Michal Feldman, Iftah Gamzu, Renato Paes~Leme, and Moshe
  Tennenholtz.
\newblock Signaling schemes for revenue maximization.
\newblock In {\em Proceedings of the 13th ACM Conference on Electronic Commerce
  (EC)}, pages 514--531, 2012.

\bibitem[FGR{\etalchar{+}}13]{feldmanclique}
Vitaly Feldman, Elena Grigorescu, Lev Reyzin, Santosh Vempala, and Ying Xiao.
\newblock Statistical algorithms and a lower bound for detecting planted
  cliques.
\newblock In {\em Proceedings of the 44th ACM Symposium on Theory of Computing
  (STOC)}, pages 655--664. ACM, 2013.

\bibitem[FK03]{feigeprobable}
Uriel Feige and Robert Krauthgamer.
\newblock The probable value of the {L}ov{\'a}sz--{S}chrijver relaxations for
  maximum independent set.
\newblock {\em SIAM Journal on Computing}, 32(2), 2003.

\bibitem[FNS07]{FNS07}
Tomas Feder, Hamid Nazerzadeh, and Amin Saberi.
\newblock Approximating nash equilibria using small-support strategies.
\newblock In {\em Proceedings of the 8th ACM conference on Electronic Commerce
  (EC)}, pages 352--354. ACM, 2007.

\bibitem[FR10]{feigeron}
Uriel Feige and Dorit Ron.
\newblock Finding hidden cliques in linear time.
\newblock {\em DMTCS Proceedings}, (01):189--204, 2010.

\bibitem[GD13]{guo}
Mingyu Guo and Argyrios Deligkas.
\newblock Revenue maximization via hiding item attributes.
\newblock In {\em Proceedings of the 23rd International Joint Conference on
  Artificial Intelligence}, pages 157--163. AAAI Press, 2013.

\bibitem[GJ79]{GareyJ79}
Michael~R. Garey and David~S. Johnson.
\newblock {\em Computers and Intractability: A Guide to the Theory of
  {NP}-Completeness}.
\newblock W. H. Freeman \& Co., 1979.

\bibitem[GLS93]{ellipsoidbook}
Martin Gr{\"o}tschel, L{\'a}szl{\'o} Lov{\'a}sz, and Lex Schrijver.
\newblock Geometric algorithms and combinatorial optimization.
\newblock {\em Algorithms and Combinatorics}, 2:1--362, 1993.

\bibitem[Hir71]{hirshleifer71}
Jack Hirshleifer.
\newblock The private and social value of information and the reward to
  inventive activity.
\newblock {\em The American Economic Review}, 61(4):561--574, 1971.

\bibitem[HK11]{hazankrauthgamer}
Elad Hazan and Robert Krauthgamer.
\newblock How hard is it to approximate the best {N}ash equilibrium?
\newblock {\em SIAM Journal on Computing}, 40(1):79--91, 2011.

\bibitem[Jer92]{jerrum92}
Mark Jerrum.
\newblock Large cliques elude the metropolis process.
\newblock {\em Random Structures \& Algorithms}, 3(4):347--359, 1992.

\bibitem[JP00]{juels}
Ari Juels and Marcus Peinado.
\newblock Hiding cliques for cryptographic security.
\newblock {\em Designs, Codes and Cryptography}, 20(3):269--280, 2000.

\bibitem[Ku{\v{c}}95]{kucera95}
Lud{\v{e}}k Ku{\v{c}}era.
\newblock Expected complexity of graph partitioning problems.
\newblock {\em Discrete Applied Mathematics}, 57(2):193--212, 1995.

\bibitem[LRS10]{lehrermediation}
Ehud Lehrer, Dinah Rosenberg, and Eran Shmaya.
\newblock Signaling and mediation in games with common interests.
\newblock {\em Games and Economic Behavior}, 68(2), 2010.

\bibitem[MW82]{MW82}
Paul~R Milgrom and Robert~J Weber.
\newblock A theory of auctions and competitive bidding.
\newblock {\em Econometrica}, 50(5), 1982.

\bibitem[NY83]{NemirovskiY83}
Arkadi{\u\i}~Semenovich Nemirovski and David~Borisovich Yudin.
\newblock {\em Problem complexity and method efficiency in optimization.}
\newblock John Wiley and Sons, 1983.

\bibitem[RT02]{RT02}
T.~Roughgarden and Eva Tardos.
\newblock How bad is selfish routing?
\newblock {\em Journal of the ACM}, 49(2):236 -- 259, 2002.

\bibitem[Rub15]{rubinstein15}
Aviad Rubinstein.
\newblock Eth-hardness for signaling in symmetric zero-sum games.
\newblock {\em CoRR}, abs/1510.04991, 2015.

\bibitem[VFH15]{VassermanFH15}
Shoshana Vasserman, Michal Feldman, and Avinatan Hassidim.
\newblock Implementing the wisdom of waze.
\newblock In {\em Proceedings of the 24th International Joint Conference on
  Artificial Intelligence (IJCAI)}, pages 660--666, 2015.

\end{thebibliography}

\begin{appendix}

\section{Proof of Lemma \ref{lem:family_t}}
\label{append-family_t}
Recall that $\epsilon > 0$, $c_3\geq 10^3$, and
  $k = k(n)$ satisfies $k = \omega(\log n)$ and $k = o(\sqrt{n})$, and $r = \Theta(n/k)$.
Let $p=\frac{1}{2}$.
 
First, we proceed as in~\cite{Dughmi14}
to reduce the planted-clique
problem to the planted-clique-cover problem.
Given an instance $G$ of $\mathbf{PClique}(n,p,k)$,
  we can generate an instance $G'$ of $\mathbf{PCover}(n,p,k,r)$ by planting $r-1$
  additional random $k$-cliques into $G$ (as in step (2) of Definition~\ref{def:pcover}).
As noted in~\cite{Dughmi14}, because the cliques $S_1, \ldots, S_r$ are indistinguishable,
  recovering a constant fraction of the planted cliques from $G'$
  would recover each of $S_1, \ldots, S_r$ with constant probability.
In particular, it can recover the original planted clique with constant probability.

So our task is the following.
Given a graph $G \sim \G(n,p,k,r)$,
  fix one of the planted $k$-cliques $S \subseteq V$.
  We need to show that given a cluster $T \subseteq V$ satisfying
  $|S \cap T| \ge \epsilon |T|$ and $|S \cap T| \ge c_3 \log n$,
  we can recover $S$ with high probability. We assume that $r=\frac{5n}{k}$ in the sequel.
Our algorithm is similar to (and in fact, simpler than) the one used in~\cite{Dughmi14} to
prove a similar planted-clique recovery result (Lemma 3.5 therein). However, we need to recover
the planted clique
under a {\em much weaker} (both qualitatively and quantitatively) assumption.
In our case, the above requirements on $|S\cap T|$ allow
$|T|=\Tht(\log n)$ (which is crucial for the soundness proof in
Lemma~\ref{lem:main} to go through); in~\cite{Dughmi14}, the requirement is that
$|S\cap T|=\Omega(|S\cup T|)$ with $|S|=\w(\log^2 n)$, so that we must have $|T|=\omega(\log^2 n)$.
This difference in the magnitude of $|T|$ (and hence $|S\cap T|$) poses certain challenges
and necessitates certain key changes to the analysis in~\cite{Dughmi14}.

We use the following algorithm to recover $S$:
\begin{enumerate}
\item
Pick an arbitrary set $R$ of $c_3\log n$ vertices from $S\cap T$.
\item Let $S'$ be all the common neighbors of $R$.
\item Let $\hat S$ be the vertices in $S'$ with at least $k-1$ neighbors in $S'$.
\end{enumerate}

Since $S$ is unknown, we use the following process to simulate Step (1).
We first sample roughly $\frac{c_3 \log n}{\epsilon}$ vertices uniformly from $T$,
  and try Step (2) and (3) on every subset of $c_3 \log n$ of the sampled vertices.
The number of subsets we need to check is polynomial.
Moreover, because $|S \cap T| \ge \epsilon |T|$, with high probability, the sampled subset
of $T$ will contain $c_3\log n$ vertices from $S\cap T$, and will encounter this set of
$c_3\log n$ vertices from $S\cap T$ in our enumeration.

We partition the edges of $G$ into $E^-$ and $E^+$,
  where $E^-$ are the background edges added in Step (1) of Definition \ref{def:pcover},
  and $E^+$ are the extra clique-related edges added in Step (2) of Definition \ref{def:pcover}.
Let $E^i$ denote the edges of $S_i$.
It is easy to verify that all the nodes in $S$ will survive Step (2) and (3),
  so $S \subseteq \hat S$. We show that in fact, with high probability,
  no other vertices survive Step (2) and (3)
  through the following claims.

\begin{claim} \label{append-vSminus}
With high probability, we have $|E^-(v, S)| \le 0.6 |S|$ for all $v \notin S$.
\end{claim}

\begin{proof}
Since $|S| = \omega(\log n)$,
this follows from a straightforward application of the Chernoff bound and the union bound.
\end{proof}

\begin{claim} \label{append-vRminus}
With high probability, there are at most $c_3\log n$ vertices $v\notin R$ with
$|E^-(v,R)|\ge 0.8|R|$.
\end{claim}

\begin{proof}
Let $A=\{v\notin R: |E^-(v,R)|\geq 0.8|R|\}$. Then, $\bidens_{G^-}(R,A)\geq 0.8$.
The constant $c_3=10^3$ and $\ve=0.3$ satisfy the conditions of Lemma~\ref{dug:bidens}.
So since $|R|\geq c_3\log n$, we have $|A|<c_3\log n$ with high probability.
\end{proof}

\begin{claim} \label{append-vSplus}
With high probability, we have $|E^+(v, S)| \le 12 \log n$ for all $v\notin S$.
\end{claim}

The following lemma will be useful in proving the above claim.

\begin{lemma}[see Ex. 1.13 in~\cite{dubhashi2009}, Lemma 1.19 in~\cite{doerr11}]
\label{dep:chernoff}
Let $X_1,\ldots,X_n$ be arbitrary binary random variables. Suppose for every $i$, and
every $x_1,\ldots,x_{i-1}\in\{0,1\}$, we have
$\Pr[X_i=1\,|\,X_1=x_1,X_2=x_2,\ldots,X_{i-1}=x_{i-1}]\leq p_i$.
Let $Y_1,\ldots,Y_n$ be independent binary random variables with $\Pr[Y_i=1]=p_i$ for all
$i\in[n]$. Then, for any $M$, we can upper bound $\Pr[\sum_{i=1}^n X_i>M]$ using the
upper-tail Chernoff bound for $\Pr[\sum_{i=1}^n Y_i>M]$.

In particular, for any $\ve\in(0,1)$ and $\mu\geq\sum_{i=1}^np_i$, we have
$\Pr[\sum_{i=1}^nX_i>(1+\ve)\mu]\leq e^{-\ve^2\mu/3}$.
\end{lemma}

\begin{proofof}{Claim~\ref{append-vSplus}}
Fix $v\notin S$, and let $X$ denote the random variable $|E^+(v,S)|$.
Let $S_1,\ldots,S_{r-1}$ be the planted cliques other than $S$.
Let $I$ be the random index-set of cliques that contain $v$; that is, $I\sse[r-1]$ is such
that $v\in S_i$ for all $i\in I$, and $v\notin S_i$ for all $i\notin I$.
Notice that the events $\{i\in I\}$ for $i\in [r-1]$ are independent Bernoulli trials with
probability $\frac{k}{n}$. So we have $\Pr[|I|>6\log n]\leq\frac{1}{n^2}$.

Fix an index set $J\sse[r-1]$ with $|J|\leq 6\log n$ and consider
$\Pr[X>12\log n\,|\,I=J]$. We use $\Pr'$ and $\Ex'$ to denote probabilities and
expectations in the space where we condition on the event $I=J$.
Conditioned on $I=J$, we have $X\leq\sum_{i\in J,u\in S}Y_{i,u}$, where $Y_{i,u}$ is the
random variable indicating if $u\in S_i$.
Fix an ordering of the $Y_{i,u}$ random variables. If we consider the random variable
$Y_{i,u}$, and any realization $\sg$ of the random variables appearing before $Y_{i,u}$,
we have
$\Pr'[Y_{i,u}=1\,|\,\text{realization $\sg$ of the variables before $Y_{i,u}$}]\leq\frac{k}{n}$.
Since $\frac{|J|k^2}{n}<6\log n$, we can now use Lemma~\ref{dep:chernoff} and infer that
$\Pr'[X>12\log n]\leq e^{-\frac{6\log n}{3}}$.

Finally, we have
\begin{align*}
\Pr[X>12&\log n] = \sum_{J\sse[r-1]}\Pr[I=J]\cdot\Pr[X>12\log n\,|\,I=J] \\
& \leq\sum_{\substack{J\sse[r-1]: \\ |J|>6\log n}}\Pr[I=J]
+\sum_{\substack{J\sse[r-1]: \\ |J|\leq 6\log n}}\Pr[I=J]\cdot\Pr[X>12\log n\,|\,I=J] \\
& \leq\Pr[|I|>6\log n]+\sum_{\substack{J\sse[r-1]: \\ |J|\leq 6\log n}}\Pr[I=J]\cdot\frac{1}{n^2}
\leq\frac{2}{n^2}. \qedhere
\end{align*}
\end{proofof}

By Claims~\ref{append-vRminus} and~\ref{append-vSplus}, and since $|R|\geq c_3\log n$,
with high probability, for all but at most $c_3\log n$ nodes $v\notin S$, we have
\[
|E(v,R)|=|E^-(v,R)|+|E^+(v,R)|\leq 0.8|R|+|E^+(v,S)|\leq 0.8|R|+12\log n\leq 0.82|R|.
\]
Hence, with high probability, at most $c_3\log n$ nodes outside of $S$ survive
Step (2), i.e., $|S'\sm S|\leq c_3\log n$.

\begin{claim} \label{append-vS}
With high probability, we have $|E(v, S)| \le 0.7 |S|$ for all $v \notin S$.
\end{claim}

\begin{proof}
Since $12\log n=o(|S|)$ (for sufficiently large $n$), by Claims~\ref{append-vSminus}
and~\ref{append-vSplus}, with probability, for all $v\notin S$, we have
$|E(v, S)| = |E^-(v, S)| + |E^+(v, S)| \le 0.6 |S| + o(|S|)\le 0.7|S|$.
\end{proof}

By Claim~\ref{append-vS} and because $|S' \setminus S|\leq c_3\log n$,
with high probability, every node
$v\in S'\setminus S$ has $|E(v, S')| \le |E(v, S)| + c_3\log n\leq 0.8|S|$.
Therefore, no vertex $v \in S' \setminus S$ survives Step (3) and $\hat S = S$.

\end{appendix}

\end{document}